\documentclass[10pt,conference]{IEEEtran}
\IEEEoverridecommandlockouts

\usepackage[T1]{fontenc} % optional
\usepackage{amsmath}
\usepackage{bm} % optional

\usepackage{xspace}
\usepackage{amsfonts}
\usepackage{amssymb}
\usepackage{amsthm}
\usepackage{stmaryrd}
\usepackage{thmtools,thm-restate}
\usepackage[textsize=small]{todonotes}
\usepackage{paralist}
\usepackage{algorithm}
\usepackage[noend]{algpseudocode}
\usepackage{fancybox}
\usepackage{mathdots}
\usepackage{balance}
\usepackage[all]{xypic}

\declaretheorem{proposition}
\declaretheorem[sibling=proposition]{lemma}
\declaretheorem[sibling=proposition]{theorem}
\declaretheorem[sibling=proposition]{corollary}
\declaretheorem[sibling=proposition]{remark}
\declaretheorem[style=remark]{example}
\declaretheorem[style=remark]{claim}

\newcommand{\minu}{\triangleleft}
\newcommand{\N}{\mathbb{N}}

\newcommand{\set}[1]{\{ #1 \}}

\newcommand{\pow}[2]{\mathcal{P}_{#1}(#2)}

\newcommand{\prettyexists}[2]{\exists #1 \colon #2}
\newcommand{\prettyforall}[2]{\forall #1 \colon #2}
\newcommand{\srcs}[1]{V_{#1}}
\newcommand{\srces}[1]{\text{\sc srces}(#1)}
\newcommand{\src}[1]{\text{\sc src}(#1)}
\newcommand{\trg}[1]{\text{\sc trg}(#1)}
\newcommand{\ord}[1]{\text{\sc ord}(#1)}
\newcommand{\inn}[1]{\text{\sc in}(#1)}
\newcommand{\precn}[1]{\text{\sc pre}(#1)}
\newcommand{\outn}[1]{\text{\sc out}(#1)}
\newcommand{\ind}[1]{\text{\sc in-deg}(#1)}
\newcommand{\outd}[1]{\text{\sc out-deg}(#1)}
\newcommand{\gr}[1]{\mathcal{G}_{#1}}

\newcommand{\ndegordgre}[2]{\mathbf{N}^{#1, {\geq}#2}} %_{#5 \set{#3,#4}}}
\newcommand{\ndegordeq}[2]{\mathbf{N}^{#1, {}#2}} %_{#5 \set{#3,#4}}}

\newcommand{\T}{\mathcal{T}}
\renewcommand{\H}{\mathcal{H}}

\newcommand{\acallgamma}{\mathbf{C}}
\newcommand{\acall}{\mathbf{D}}
\newcommand{\acallpar}[1]{\mathbf{D}^{#1}}

\newcommand{\kerne}[1]{\text{\sc ker}(#1)}
\newcommand{\ndegall}{\mathbf{N}}
\newcommand{\ndegn}[1]{\mathbf{N}^{{\geq}#1}}

\newcommand{\ap}[3]{\mathbf{P}_{\set{#1, #2} #3}}

\newcommand{\apaX}[2]{\mathbf{C}^{#1}_{#2}}

\newcommand{\zerovector}{\mathbf{0}}
\newcommand{\supp}[1]{\text{\sc supp}(#1)}

\newcommand{\para}[1]{\medskip\subsubsection*{\bf #1}}
\newcommand{\parainproof}[1]{\medskip\subsubsection*{\emph{\em #1}}}
\makeatletter
\def\BState{\State\hskip-\ALG@thistlm}
\makeatother

\algrenewcommand{\algorithmiccomment}[1]{\qquad$\rightarrow$ #1}

\newcommand{\tuple}[1]{\langle #1 \rangle}

\newcommand{\atoms}{\text{\sc Atoms}}
\newcommand{\A}{\mathcal{A}}
\newcommand{\G}{\mathcal{G}}
\newcommand{\C}{\mathcal{C}}
\newcommand{\X}{\mathcal{X}}

\newcommand{\setof}[2]{\set{#1 \,  : \, #2}}

\newcommand{\kNRA}[1]{\ensuremath{#1\textsf{-\sc NRA}}\xspace}

\newcommand{\kCFG}[1]{\ensuremath{#1\textsf{-\sc CFG}}\xspace}
\newcommand{\conf}[2]{{#1}(#2)}

\newcommand{\trans}[1]{ \xrightarrow{#1}  }

\newcommand{\Par}[1]{\text{\sc Par}(#1)}
\newcommand{\emptyword}{\varepsilon}
\newcommand{\langof}[1]{L(#1)}
\newcommand{\langoffull}[5]{L_{\conf {#2}{#3}\, \conf{#4}{#5}}(#1)}

\newcommand{\dom}[1]{\text{\sc dom}(#1)}
\newcommand{\size}[1]{|#1|}
\newcommand{\lang}[1]{L(#1)}
\newcommand{\langpar}[2]{L_{#1}(#2)}
\newcommand{\langh}[1]{H_{#1}}
\newcommand{\langw}[1]{W_{#1}}
\newcommand{\yield}[1]{\text{\sc yield}(#1)}

\newcommand{\bez}[2]{#1 - #2}
\newcommand{\poj}[1]{\text{\sc Sing}(#1)}

\newcommand{\pair}[2]{\langle #1,#2\rangle}

\newcommand{\restr}[2]{{#2}_{|#1}}

\newcommand{\Gr}{\text{\sc Perm}}

\ifCLASSINFOpdf
\else
\fi
\begin{document}

\title{Parikh's theorem for infinite alphabets%
\thanks{This work was partially supported by NCN grants 2016/21/D/ST6/01368, 2017/27/B/ST6/02093  and
2019/35/B/ST6/02322.}
}


\author{
\IEEEauthorblockN{Piotr Hofman}
\IEEEauthorblockA{University of Warsaw}
\and
\IEEEauthorblockN{Marta Juzepczuk}
\IEEEauthorblockA{University of Warsaw}
\and
\IEEEauthorblockN{S{\l}awomir Lasota}
\IEEEauthorblockA{University of Warsaw}
\and
\IEEEauthorblockN{Mohnish Pattathurajan}
\IEEEauthorblockA{University of Warsaw}
}


%Telephone: (800) 555--1212, Fax: (888) 555--1212}


\IEEEoverridecommandlockouts
\IEEEpubid{\makebox[\columnwidth]{978-1-6654-4895-6/21/\$31.00~
\copyright2021 IEEE \hfill} \hspace{\columnsep}\makebox[\columnwidth]{ }}

\maketitle

\begin{abstract}
We investigate commutative images of languages recognised by register automata and grammars.
Semi-linear and rational sets can be naturally extended to this setting by allowing for orbit-finite unions 
instead of only finite ones.
We prove that commutative images of languages of  one-register automata are not always semi-linear, 
but they are always rational.
We also lift the latter result to grammars: commutative images of one-register context-free languages
are rational,
and in consequence commutatively equivalent to register automata.
We conjecture analogous results for automata and grammars with arbitrarily many registers.
\end{abstract}

%\thanks{This work was supported by NCN grant UMO-2016/21/D/ST6/01368.}

% !TEX root = main.tex

%
%\begin{enumerate}
%\item
%Define automata with $k$ pushdown stacks ($k$PDA), for $k\geq 1$,
%and the language recognised by a $k$PDA.
%For $k > 2$, define a reduction from the non-emptiness problem for $k$PDA
%to the non-emptiness problem for $2$PDA.\\
%\item
%Let $L\subseteq A^*$ be a fixed regular language over an alphabet $A$.
%Let $K\subseteq A^*$ be the smallest language satisfying the following conditions:
%\begin{itemize}
%\item $\varepsilon \in K$
%\item $K \cdot K \subseteq K$
%\item $u \cdot K \cdot v \subseteq K$ for every words $u, v\in A^*$ such that $u v \in L$.
%\end{itemize}
%Is $K$ always regular? Is $K$ always context-free?
%\end{enumerate}

\section{Introduction}

\emph{Register automata}, introduced over 25 year ago by Francez and Kaminski~\cite{KF94},
are nondeterministic finite-state devices equipped with a finite number of
registers that can store data values from an infinite data domain.
A register automaton inputs a string of data values (a data word) and
%When processing an input string, an automaton 
compares each consecutive input data value to its registers; 
based on this comparison and on the current control state,
it chooses a next control state and possibly stores the input value in one of its registers.
%As in the original model of Francez and Kaminski, 
The only allowed comparisons of data values considered in this paper are equality tests.
An automaton can also guess a fresh data value not previously seen in the input, and store it in a register
(we thus consider nondeterministic register automata \emph{with guessing}~\cite{S06}).
Likewise one defines register context-free grammars~\cite{lmcs14}, \cite[Sect.5]{atombook}. 

Register automata lack most of the good properties known from the classical theory of finite automata,
like determinisation or closure properties.
In particular, no satisfactory characterisation in terms of rational (regular) expressions is known.
Indeed, all known generalisations of Kleene's theorem for register automata
either apply to a restricted subclass of the model,
or introduce an involved syntax significantly extending the concept 
of rational expressions~\cite{regexp-Kaminski,regexp-Domagoj,regexp-Kurz}. 

Register automata are expressively equivalent to \emph{orbit-finite automata}~\cite{lics11,lmcs14},
a natural extension of finite automata where one allows for input alphabets and state spaces which are infinite, but finite 
up to permutation of the data domain (= orbit-finite).
Along the same lines, 
in this paper we focus on a natural extension of rational expressions, which differ from the classical ones
just by allowing for \emph{orbit-finite unions}.
In other words, we consider the class of \emph{rational languages}, defined as
the smallest class of languages closed under concatenation, star,
%iteration $\_^*$,
and orbit-finite unions. In particular, the class contains the empty language, all finite and all orbit-finite languages.

Languages of register automata are not rational in general, even in case of deterministic one-register
automata.
%, as shown in Section~\ref{sec:rat}.
%Example~\ref{ex:L1L2} below.
Kleene theorem may be however recovered when commutative images
(Parikh images) are considered:
we prove that
the language of every one-register automaton is Parikh-equivalent to (i.e., has the same Parikh image as) 
a rational language.
%Example~\ref{ex:L1L2} illustrates.
%
\begin{example} \label{ex:L1L2}
Fix the data domain $\atoms = \set{0,1,2,\ldots}$ and
consider the language $L_1$ 
%over the alphabet $\N$
consisting of all nonempty words over $\atoms$ where every two consecutive letters are different:
\[
L_1 \ = \ \setof{ a_1 a_2 \ldots a_n \in \atoms^* }{a_1 \neq a_2 \neq \ldots \neq a_n}.
\]
The language is recognised by a deterministic one-register automaton but it is not rational
(cf.~Section~\ref{sec:rat}).
It is however Parikh-equivalent to a larger language $L_2$, where
the non-equality constraint is imposed at every second position only:
\[
L_2 \ = \ \setof{ a_1 a_2 \ldots a_n \in \atoms^* } {a_1 \neq a_2, \ a_3 \neq a_4, \ \ldots},
\]
which is defined by the rational (regular) expression
\begin{align} \label{eq:L2rat}
L_2 \ = \ \Big(\bigcup_{a,b \in\atoms, a\neq b} a b \Big)^* \, \big(\emptyword  \,\, \cup \, \bigcup_{a\in\atoms} a\big)
\end{align}
and is thus rational
(the formal definition of rational languages will be given in Section~\ref{sec:rat}).
Indeed, every  $w \in L_2$ can be transformed, by swapping letters, to a word in $L_1$.
Let $w = a_1 a_2 \ldots a_{n}$. % (the case of even length is similar).
If $a_{2} = a_{3}$ we swap non-equal letters $a_{3}$ and $a_{4}$ 
thus achieving $a_{1} \neq a_{2} \neq a_{3} \neq a_4$.
Next, if $a_{4} = a_{5}$ we swap analogously $a_{5}$ and $a_{6}$, and so on.
Continuing in this way 
%thus achieving $a_{2n-3} \neq a_{2n-2} \neq a_{2n-1} \neq a_{2n} \neq a_{2n+1}$.
%And so on, until
%finally, if $a_2 = a_3$, we swap $a_1$ and $a_2$ thus 
we finally arrive at a word in $L_1$.
\qed
\end{example}

\para{Contribution}
We contribute to understanding commutative images of languages of one-register automata and grammars,
by investigating sets of \emph{data vectors} obtainable as Parikh images of these languages.
Parikh images of rational languages we call rational as well.
Here are our contributions:
\begin{itemize}
\item[(1)] 
We show that Parikh images of languages of one-register automata are not \emph{semilinear sets} of data vectors in general.
By seminilinear sets we naturally mean orbit-finite unions of linear sets, 
which in turn are determined by a base and an orbit-finite set of periods, like classically.
%Thus Parikh images of register automata are not semilinear in general.
\item[(2)]
We prove that languages of one-register automata have rational Parikh images.
The crucial part of the proof resorts to a graph-theoretical characterisation of these Parikh images,
and uses a necessary condition for a Hamiltonian cycle in directed graphs
(\cite{Ghouila-Houri,Kuhn-Osthus}).
\item[(3)]
Finally, we extend (2) to context-free grammars by showing that
one-register context-free languages have rational Parikh images.
The result is obtained by a novel type of transformations of derivation trees.
\end{itemize}

We conjecture that the restriction on the number of registers in (2) and (3) can be dropped; 
the combinatoric complexity we have encountered already in one-register case
makes it however difficult to envisage a generalisation of our approach to the general case.
According to (1), one-register automata and grammars fail to have semilinear Parikh images in general.
However, as a direct corollary of (3) we recover an analog of the Parikh's classical theorem~\cite{Parikh66}:
%
%\begin{theorem} \label{cor:Parikh}
one-register context-free grammars are 
Parikh-equivalent to register automata (but not to one-register automata).
%\end{theorem}
%
\para{Related research}
Register automata have been intensively studied with respect to their foundational 
properties~\cite{KF94,SI00,regexp-Kaminski,NSV04}. 
Following the seminal paper of Francez and Kaminski~\cite{KF94}, 
subsequent extensions of the model allow for
comparing data values with respect to some fixed relations such as a total order, or introduce alternation, variations on the allowed form of nondeterminism, etc.
The model is well known to satisfy almost no semantic equivalences that hold for classical finite automata.
%in particular register automata admit no satisfactory characterizations 
%in terms of regular expressions~\cite{regexp-Kaminski,regexp-Domagoj} or logic~\cite{NSV04,CM14}.
Here are few positive results: simulation of two-way nondeterministic automata
by one-way alternating automata with guessing~\cite{atombook}; 
Myhill-Nerode-style characterisation of languages of deterministic 
automata~\cite{MN-Kaminski,lics11,lmcs14}; and the well-behaved class of languages definable by 
orbit-finite monoids~\cite{datamonoids}, characterised in terms of logic~\cite{rigidMSO}
and a syntactic subclass of deterministic automata~\cite{BS20}.
Register automata have been also intensively studied with respect to  
their applications to XML databases and logics~\cite{DL09,NSV04,CM14} (see~\cite{S06} for a survey).
Register context-free grammars are equivalent to register pushdown automata~\cite{atombook,CL15}.

Other extensions of finite-state machines to infinite alphabets include:
abstract reformulation or register automata, known as orbit-finite automata, or nominal automata, or automata over atoms)~\cite{lics11,lmcs14,atombook}; symbolic automata~\cite{symbautom};
pebble automata~\cite{MSV03}; and data automata~\cite{dataautomata,classautomata}.

%These extensions found applications in various contexts, notably in investigations on finite-state manipulation of semi-structured data~\cite{}.

% !TEX root = main.tex

 \section{Preliminaries}    \label{sec:pre}

\para{Sets with atoms} 

Our definitions rely on basic notions and results of the theory of \emph{sets with atoms}~\cite{atombook}, 
also known as nominal sets~\cite{Pitts:book}. 
This paper is a 
%In this section we recall what is necessary to follow our arguments; this is 
part of a uniform 
abstract approach to register automata in the realm of orbit-finite sets with atoms, developed in~\cite{lics11,lmcs14,atombook}.

Fix a countably infinite set $\atoms$, whose elements we call \emph{atoms}.
We reserve initial alphabet letters $a,b,  \ldots$ to range over atoms.
Informally speaking, a set with atoms is a set that can have atoms, or other sets with atoms,
as elements.
Formally, we define the universe of sets with atoms by a suitably adapted cumulative hierarchy of sets,
by transfinite induction: % on ordinals $\gamma$:
the only set of \emph{rank} 0 is the empty set; and for a cardinal $\gamma$, a set of rank $\gamma$ may contain, as elements, 
sets of rank smaller than $\gamma$ as well as atoms.
In particular, nonempty subsets $X\subseteq\atoms$ have rank 1.

Denote by $\Gr$ the group of all permutations of $\atoms$.
Atom permutations $\pi: \atoms\to\atoms$ act on sets with atoms by consistently renaming all atoms in
a given set. 
Formally, by another transfinite induction we define $\pi(X) = \setof{\pi(x)}{x\in X}$.
Via standard set-theoretic encodings of pairs or finite sequences we obtain, in particular, 
the pointwise action on pairs $\pi (x,y)=(\pi(x),\pi(y))$,
and likewise on finite sequences.  %: $\pi\cdot (x_1, \ldots, x_n)=(\pi\cdot x_1, \ldots, \pi\cdot x_n)$ .
Relations and functions from $X$ to $Y$ are considered as subsets of $X\times Y$;
for instance, in case of $f:\atoms\to\atoms$, we have $\pi(f)(a)=\pi(f(\pi^{-1}(a)))$.
%If a set $X$ is \emph{pure}, i.e., no atom belongs to $X$, its elements, elements of its elements, etc., 
%then $\pi(X) = X$ for every $\pi\in \Gr$; for instance $\pi(n) = n$ for $n\in\N$.

We restrict to sets with atoms that only depend on finitely many atoms, in the following sense. 
A \emph{support} of $x$ is any set $S\subseteq\atoms$ such that the following implication holds for all $\pi\in \Gr$:
%\begin{quote}
    if $\pi(s)=s$ for all $s\in S$, then $\pi(x)=x$.
%\end{quote}
An element (or set) $x$ is \emph{finitely supported} if it has some finite support; in this case
%Finite supports of a fixed element are closed under intersections, and hence every finitely supported 
$x$ has {\em the least support}, denoted $\supp x$, called \emph{the support} of $x$ 
(cf.~\cite[Sect.~6]{atombook}).
%(c.f.~\cite[Prop.~2.3]{Pitts:book}, \cite[Cor.~9.4]{lmcs14}).
Sets supported by $\emptyset$ we call \emph{equivariant}.
For instance, given $a,b\in\atoms$, the support of the set
\[
L_{ab} \ = \ \setof{a_1 a_2 \ldots a_n \in \atoms^*}{n\geq 2, \ a_1 \neq a, \ a_n = b}
\]
is $\set{a,b}$;
% (indeed, for every $\pi\in \Gr$ such that $\pi(a) = a$ and $\pi(b) = b$, we have $\pi(L) = L$).
the projection function $\pi_1 : \atoms^2 \to \atoms : \pair a b \mapsto a$
is equivariant; 
% (indeed, every $\pi \in \Gr$ preserves the graph $\setof{\pair{\pair a b} a}{a,b\in\atoms}$ of $\pi_1$);
the support of a sequence $\tuple{a_1 \ldots a_n}\in\atoms^*$, encoded as a set in a standard way, is
the set of atoms $\set{a_1, \ldots, a_n}$ appearing in it;
and the support of a function $f: \atoms\to\N$ such that
$\dom f = \setof{a\in\atoms}{f(a)>0}$ is finite, is exactly $\dom f$.

From now on, we shall only consider sets with atoms that are hereditarily finitely
supported (called briefly \emph{legal}), i.e., ones that have a finite support, whose every element has some finite support,
and so on.
%A set $X$ in the cumulative hierarchy is called \emph{legal} (or \emph{set with atoms})
%if it is hereditarily finitely-supported: it is finitely-supported itself, all its elements are finitely supported, etc.
%From now on all sets are silently assumed to be legal.

\para{Orbit-finite sets} 

Two (elements of) sets with atoms $x, y$  are \emph{in the same orbit} if $\pi(x) = y$ for some $\pi \in \Gr$.
This equivalence relation splits every set with atoms $X$ into equivalence classes, which we call \emph{orbits in}  $X$.
A (legal) set is \emph{orbit-finite} if it splits into finitely many orbits.
Examples of orbit-finite sets are: $\atoms$ (1 orbit); 
$\atoms - \set{a}$ for some $a\in\atoms$ (1 orbit);
$\atoms^2$ (2 orbits: diagonal and non-diagonal);
$\atoms^3$ (5 orbits, corresponding to equality types of triples);
$\set{1,\ldots, n} \times \atoms$
($n$ orbits, as $\pi(i) = i$ for every $i\in\N$ and $\pi \in \Gr$, according to the standard set-theoretic definition of natural numbers);
%atom permutations fix elements of $H$);
the set of $n$-element subsets of atoms $\pow n \atoms = \setof{X\subseteq \atoms}{\size X = n}$ (1 orbit).
%
%This paper is an attempt to generalize Parikh's theorem from finite to orbit-finite alphabets.

Given a  family $(X_i)_{i\in I}$ of sets indexed by an orbit-finite set $I$,
the union $\bigcup_{i\in I} X_i$ we call \emph{orbit-finite union} of sets $X_i$.
(Formally, not only each set $X_i$ is assumed to be legal, but also the indexing function
$i\mapsto X_i$.)
As an example, consider $(L_{ab})_{b\in\atoms}$.
The indexing function $b\mapsto L_{ab}$ is supported by $\set{a}$, and
so is the union:
\[
\bigcup_{b\in\atoms}  L_{ab} \ = \ 
\setof{a_1 a_2 \ldots a_n \in \atoms^*}{n\geq 2, \ a_1 \neq a}.
\]
Orbit-finite sets are closed under Cartesian products, subsets, and orbit-finite unions:
if each of $X_i$ is orbit-finite, their union $\bigcup_{i\in I} X_i$ is orbit-finite 
too~\cite[Sect.~3]{atombook}.
%%
%\begin{lemma}[\todooo{citation}\cite{}] \label{lem:o-f}
%Orbit-finite union of orbit-finite sets is orbit-finite.
%\end{lemma}

\para{Data words and vectors}
By a finite multiset over a set $\Sigma$ 
we mean any function $v : \Sigma \to \N$ such that $v(\alpha) = 0$ for all $\alpha\in \Sigma$ except finitely many.
We define the \emph{domain} of $v$ as $\dom{v} = \setof{\alpha\in\Sigma}{v(\alpha)>0}$, and
its \emph{size}  as $\size v = \sum_{\alpha\in \dom v} v(\alpha)$
(the same notation is used for the size of a set).
%such that $v(\alpha) = 0$ for almost all $\alpha\in\Sigma$.
%
The \emph{Parikh image} (commutative image)
of a word $w\in \Sigma^*$ is the multiset $\Par{w} : \Sigma\to\N$, where
$\Par{w}(\alpha)$ is the number of appearances of a letter $\alpha\in\Sigma$ in $w$.
For a language $L \subseteq \Sigma^*$, its Parikh image is
$\Par L = \setof{ \Par w }{w\in L}$.
Two languages $L, L'\subseteq \Sigma^*$ are \emph{Parikh-equivalent} if they have the same Parikh images:
$\Par{L} = \Par{L'}$.
%Two devices recognizing data languages (e.g.~register automata or register grammars) 
%are Parikh-equivalent if their languages are so, and they are \emph{equivalent} if
%their languages are equal.
We write $\size w$ for the \emph{length} of $w$, hence
$\size v = \size w$ if $v = \Par w$.
We order multisets pointwise: $v\sqsubseteq v'$ if $v(\alpha)\leq v'(\alpha)$ for all $\alpha\in\Sigma$.
The zero (empty) multiset $\zerovector$ satisfies $\zerovector(\alpha) = 0$ for every $\alpha\in\Sigma$.
A singleton, written $\set{\alpha}$, maps $\alpha$ to $1$ and all other letters to $0$.
Addition of multisets is pointwise: $(v + v')(\alpha) = v(\alpha) + v'(\alpha)$ for every $\alpha\in\Sigma$;
likewise subtraction $v-v'$, for $v' \sqsubseteq v$. 

When $\Sigma$ is an orbit-finite alphabet,
words $w\in\Sigma^*$ we traditionally call \emph{data words}, languages $L \subseteq \Sigma^*$
%of data words 
we call \emph{data languages}, and
finite multisets $v : \Sigma\to \N$ we call \emph{data vectors}.
Orbit-finiteness of a set of data words (or data vectors) is equivalent to bounded length (or size) of its elements:
\begin{lemma}%[Appendix~\ref{sec:pre-app}] 
\label{lem:b-o-f}
A set $X$ of data words or data vectors over an orbit-finite alphabet $\Sigma$ is orbit-finite if, and only if,
$\setof{\size v}{v\in X}\subseteq \N$ is bounded.
\end{lemma}

\para{One-register automata}
For defining register automata we consider input alphabets of the form $\Sigma = H\times \atoms$, 
where $H$ is a finite set. 
We use three fixed variables $x,y,x'$ to represent register values and input atoms.
A \emph{nondeterministic register automaton} with one register (\kNRA 1) $\A$ consists of: 
%\begin{itemize}
%\item a positive integer $r\in\N$ (the number of registers),
a finite set $H$ (finite component of the alphabet),
a finite set of control locations $Q$, 
subsets $I, F \subseteq Q$ of initial resp.~accepting locations, and
a finite set $\Delta$ of transition rules of the form
\begin{align} \label{eq:trrule}
(\conf q x, \pair h y, \varphi, \conf {q'}{x'})
\end{align}
\noindent
where $q,q'\in Q$, $h\in H$, and  
$\varphi(x,y,x')$ is a Boolean combination of equalities involving the variables $x,y,x'$,
%drawn as:
%%
%\vspace{-3mm}
%\begin{figure}[H] %[btp]
%\begin{center}
%\includegraphics[width=2.2cm]{PICS/trrule.png}
%\end{center}
%\end{figure}
%\vspace{-5mm}
%%
%
%%\end{itemize}
%%
specifying relation between current register value ($x$), input atom ($y$),
and next register value ($x'$) resulting from a transition.

%We write $\atoms^{(r)}$ for the set of non-repeating $r$-tuples of atoms.
A configuration $\pair q a \in Q\times \atoms$ of $\A$, written $\conf q a$, 
consists of a control location $q\in Q$ and a register value
$a \in \atoms$.
%It is initial (resp.~accepting) if location $q$ is so, irrespectively of $a$.
For all atoms $a,b,a'$ such that $(a, b, a') \models \varphi$, a rule~\eqref{eq:trrule} induces a transition 
\[
 {\conf q {a}} \trans{\tuple{h,b}} {\conf {q'} {a'}}
 \]
from a configuration $\conf q {a}$ to a configuration $\conf {q'} {a'}$.
The semantics of \kNRA 1 is defined as in case of classical NFA, with configurations considered as states
and $\Sigma = H\times\atoms$ as an alphabet.
A \emph{run} of $\A$ over a data word $w=\tuple{h_1,b_1} \tuple{h_2,b_2} \ldots \tuple{h_n,b_n} \in \Sigma^*$ is any sequence
\begin{align} \label{eq:run}
{\conf {q_0} {a_0}} \trans{\tuple{h_1,b_1}} {\conf {q_1}{a_1}} \trans{\tuple{h_2,b_2}}\ldots 
\trans{\tuple{h_n,b_n}} \conf{q_n}{a_n}.
\end{align}
%
% where $q_0 \in I$;
 %$\conf {q_0}{a_0}$ is an initial configuration; 
 %it is \emph{accepting} if $q_n\in F$.
% the ending configuration  $\conf {q_n}{a_n}$ is accepting.
%A data word $w$ is \emph{accepted} by $\A$ if there is an accepting run of $\A$ on $w$. 
%
Let $\langoffull \A q a {q'} {a'}$ be the set of data words admitting a run starting
in $\conf {q_0} {a_0} = \conf q a$ and ending in $\conf {q_n} {a_n} = \conf {q'} {a'}$.
The \emph{language} recognised by $\A$, denoted $\lang{\A}$, is defined as:
\begin{align} \label{eq:L}
\langof \A \ = \ \bigcup_{q\in I, q'\in F, a,a'\in \atoms} \langoffull \A q a {q'} {a'}.
\end{align}

\begin{remark} \rm
The definition allows for \emph{guessing}, i.e., an automaton may nondeterministically guess, and store
in its register, an atom not yet seen in the input (cf.~\cite{S06}). 
In particular, the initial register value is guessed nondeterministically.
\end{remark}

\begin{example} \label{ex:1NRA}
Let $H$ be a singleton, omitted below;
we thus consider $\atoms$ as an alphabet.
The \kNRA 1 consisting of  $Q = F = \set{q,p}$, $I = \set{q}$, and two transition rules:
\begin{align*}
(\conf q x, y, y = x', \conf {p}{x'})
\qquad
(\conf p x, y, x\neq y = x', \conf {p}{x'})
\end{align*}
recognises $L_1$ from Example~\ref{ex:L1L2}, and can be drawn as:
\vspace{-2mm}
\begin{figure}[H] %[btp]
\begin{center}
\includegraphics[width=4.8cm]{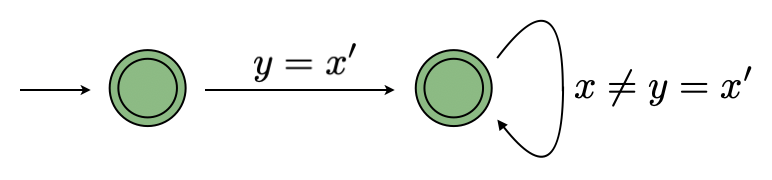}
\end{center}
\end{figure}
\vspace{-4mm}
%
%\noindent
%The automaton is  deterministic except for the initial guessing.
\end{example}
%
%\[
%\xymatrix{
%{}\ar[r]&
%%*+[o][F-]{p}\ar@/^/[rr]^(.2){\top}^(.8){\downarrow 1}\ar@(dr,dl)[]^(.2){\top}  & & 
%*+<9pt>[o][F=]{\phantom{a}} 
%\ar[rr]^(.5){y=x'} %\ar@(dr,dl)[]^(.3){x_1\neq x = x'_1}
%%\ar@/_/[rr]_(.4){x_1\neq x=x'_1}
%& & 
%*+<9pt>[o][F=]{\phantom{a}}
%\ar@(ur,dr)[]^(.5){x\neq y=x'}
%%\ar@/^/[ll]^(.2){x\neq x_1}
%}
%\]
%\qed

%The above restrictions define a \emph{normal form} of \NRA.

%\begin{claim}
%Every \NRA can be transformed into an equivalent \NRA in  normal form with the same number of registers.
%\end{claim}
%
%In the sequel we always silently assume that \NRA are in normal form.
%
%
%
%\begin{lemma}\label{GuesstoNo}
%Every \NRA with guessing is Parikh-equi\-va\-lent to a \NRA without guessing
%with one more register.
%\end{lemma}
%%
%\begin{proof}
%Consider a \NRA $\A$ and assume, without loosing generality,
%that every atom guessed and stored in some register is eventually seen later in the input.
%Indeed, a fresh atom which is never seen in the input is useless, and \NRA can use
%nondeterminism to guess the uselessness in advance and avoid guessing at all.
%
%The simulating \NRA $\A'$ without guessing, upon a transition of $\A$ that guesses a fresh atom
%and stores it in register $i$,
%performs a dummy transition that just reads an atom from the input and stores it in register $i$.
%
%In addition, locations of $\A'$ maintain, for each register $i$, 
%the request that atom residing in register $i$ still has to be seen in the input.
%This request is released when a transition of $\A$ is simulated containing the constraint $x_i = x$. 
%
%\todooo{todo: formalize}
%\end{proof}

\para{One-register context-free grammars}

For technical convenience we restrict to production rules of arity at most 2 
(higher arities can be treated similarly, but inessentially increase the combinatorial complexity of
Section~\ref{sec:anti}, see the comment in Section~\ref{sec:conc}). 
Unary production rules are easily simulated using binary and nullary ones. 

A \emph{context-free grammar} with one register (\kCFG 1) $\G$ consists of: 
%\begin{itemize}
%\item a positive integer $r\in\N$ (the number of registers),
two finite sets $H$ and $Q$ of terminals and nonterminals,
an initial nonterminal $q_0 \in Q$,
and two finite sets	 $\Delta_2$ and $\Delta_0$ of binary and nullary \emph{production rules}, of the forms
\begin{align} \label{eq:trrulecfg}
\conf q x \trans{\varphi} \conf p y\,\, \conf {p'}{y'} \ \in \ \Delta_2, \qquad \conf q x \trans{} \varepsilon \ \in \ \Delta_0,
%(q, \varphi, p, p') \in \Delta_2, \qquad (q, h, \psi) \in \Delta_0,
\end{align}
where $q\in Q$, $p,p'\in Q\cup H$, and
$\varphi(x, y, y')$ is a Boolean combination of equalities involving the three (still fixed) variables.
%We draw the production rules as:
%
%\vspace{-3mm}
%\begin{figure}[h] %[btp]
%\begin{center}
%\includegraphics[width=5.5cm]{PICS/cfg-rules.png}
%\end{center}
%\end{figure}
%\vspace{-5mm}
%
Similarly as before, a configuration $\conf q a \in Q\times \atoms$ of $\A$ 
consists of a nonterminal $q\in Q$ and a register value $a \in \atoms$.
Elements of $\Sigma = H\times\atoms$ we denote either as $\conf h a$ or as $\pair h a$.
%It is initial if nonterminal $q$ is so, irrespectively of $a$.
Production rules~\eqref{eq:trrulecfg} induce \emph{productions} 
\begin{align} \label{eq:prod}
 {\conf q {a}} \trans{} {\conf {p} {b}} \,\, {\conf {p'}{b'}},
 \qquad
 {\conf q {a}} \trans{} {\varepsilon},
 \end{align}
the former one under the condition $(a,b,b')\models \varphi$.
We denote by $\Pi_2$ and $\Pi_0$, respectively, the (infinite) sets of productions induced by
rules from $\Delta_2$ and $\Delta_0$.
 
The semantics of \kCFG 1 is defined as for classical CFG,
with configurations considered as nonterminals, alphabet $\Sigma = H\times\atoms$,
and productions $\Pi_2 \cup \Pi_0$.
Derivation trees $\T$ of $\G$ are labeled by configurations, alphabet letters $\pair h a = \conf h a\in \Sigma$, 
or the empty word $\varepsilon$, in a way consistent with productions~\eqref{eq:prod}:

\vspace{-2mm}
\begin{figure}[H] %[btp]
\begin{center}
\includegraphics[width=4.4cm]{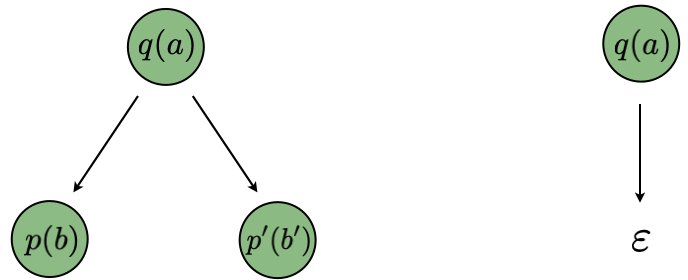}
\end{center}
\end{figure}
\vspace{-3mm}

\noindent
\emph{Complete} derivation trees have all leaves labeled by elements of $\Sigma \cup \set{\varepsilon}$.
We write $\langpar {\conf q a} \G\subseteq \Sigma^*$ for the language of  
 yields of all complete derivation trees $\T$ with root labeled by $\conf q a$,
as usual, where
$\yield \T \in \Sigma^*$ is obtained as concatenation of labels of the leaves of $\T$.
The language $\lang \G\subseteq \Sigma^*$ generated by $\G$ is defined as the union
(as in case of \kNRA 1, the initial register value is guessed nondeterministically):
\[
\lang \G \ = \ \bigcup_{a \in\atoms} \langpar {\conf {q_0} a} \G.
\]
\begin{example} \label{ex:1G}
The \kCFG 1 consisting of
nonterminals $Q = \set{q,p}$, terminals $H = \set{l,r}$, 
initial nonterminal $q$, and rules
\begin{align*}
\conf q x \trans{x \neq  y = y'} \pair l y  \conf {p} {y'} \ \ 
\conf {p} {x}  \trans{x = y = y'} \conf q y  \pair r {y'} \ \ 
 \conf q x \trans{} \varepsilon
\end{align*}
generates palindrome-like words of the form
\[
\pair l {a_1} \pair l {a_2} \ldots \pair l {a_n} \,
\pair r {a_n} \ldots \pair r {a_2} \pair r {a_1} 
\]
where $n\geq 0$ and $a_1 \neq a_2 \neq \ldots \neq a_n$.
\end{example}

\begin{remark} \rm
An alphabet $H\times\atoms$ and configurations $Q\times\atoms$ are orbit-finite. 
\kNRA 1 and \kCFG 1 are thus special cases of 
the abstract notions of orbit-finite automata and context-free grammars 
(cf.~\cite[Sect.~5]{atombook}),
where alphabets, state spaces and nonterminals may be  arbitrary orbit-finite sets.
\end{remark}

\para{Normal forms}
In the sequel we assume,
w.l.o.g., that each constraint $\varphi$ appearing 
in a transition rule~\eqref{eq:trrule} 
defines a single orbit of $\atoms^{3}$. 
In other terms, $\varphi$ contains either equality or disequality of  every pair of variables.
This can be easily achieved by splitting every constraint into a number of single-orbit ones.
For the automaton from Example~\ref{ex:1NRA} we get:
\vspace{-1mm}
\begin{figure}[H] %[btp]
\begin{center}
\includegraphics[width=4.8cm]{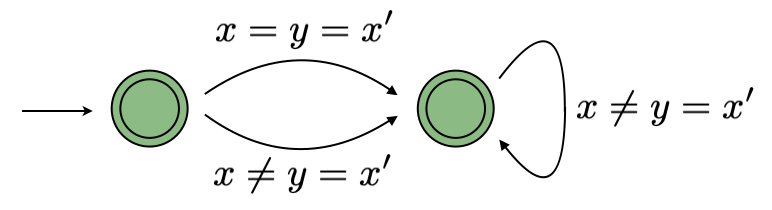}
\end{center}
\end{figure}
\vspace{-3mm}
%
%\[
%\xymatrix{
%{}\ar[r]&
%%*+[o][F-]{p}\ar@/^/[rr]^(.2){\top}^(.8){\downarrow 1}\ar@(dr,dl)[]^(.2){\top}  & & 
%*+<9pt>[o][F=]{\phantom{a}} 
%\ar@/^/[rr]^(.5){x=y=x'} %\ar@(dr,dl)[]^(.3){x_1\neq x = x'_1}
%\ar@/_/[rr]_(.5){x\neq y=x'}
%& & 
%*+<9pt>[o][F=]{\phantom{a}}
%\ar@(ur,dr)[]^(.5){x\neq y=x'}
%%\ar@/^/[ll]^(.2){x\neq x_1}
%}
%\]

%\noindent
There are thus just five possible constraints $\varphi$ and, correspondingly, five \emph{types} of transition rules. 
The first two types preserve register value ($x = x'$):
\newcommand{\ee}{{\small $\langle 1\rangle$}\xspace}
\newcommand{\nne}{{\small $\langle 2\rangle$}\xspace}
\renewcommand{\ne}{{\small $\langle 3\rangle$}\xspace}
\newcommand{\en}{{\small $\langle 4\rangle$}\xspace}
\newcommand{\nnn}{{\small $\langle 5\rangle$}\xspace}
\newcommand{\fee}{\varphi_1}
\newcommand{\fnne}{\varphi_2}
\newcommand{\fne}{\varphi_3}
\newcommand{\fen}{\varphi_4}
\newcommand{\fnnn}{\varphi_5}
\begin{enumerate} 
\item[\ee] $\fee \equiv x = x' = y$ (register value equal to input atom);
\item[\nne] $\fnne \equiv x = x' \neq y$ (register value different from input).
\end{enumerate}
The remaining types describe an update of register value:
\begin{enumerate}
\item[\ne] $\fne \equiv x \neq y = x'$ (register updated with input atom); %register value is different from input atom, and updated with this atom);
\item[\en] $\fen \equiv x = y \neq x'$ (register updated freshly); % value is equal to input atom );
\item[\nnn] $\fnnn \equiv x \neq y \neq x' \neq x$ (register updated freshly). % value is different from input atom, ).
\end{enumerate}
In the sequel we distinguish between \emph{register-preserving}
(types \ee, \nne) and \emph{register-updating} (types \ne--\nnn) constraints $\varphi$, 
transition rules, and transitions.

Likewise we assume, w.l.o.g., that each constraint
$\varphi$ appearing in a production rule of a \kCFG 1 defines a single orbit of 
$\atoms^3$.
%Thus there are just five types of rules in $\Delta_2$. For instance, 
The grammar in Example~\ref{ex:1G} is in normal form.

% !TEX root = main.tex
 
\section{Rational sets} \label{sec:rat}

%We investigate below sets of data vectors which are Parikh images of languages of \kNRA 1. 
In this section we define rational sets of data words and data vectors, 
prove their closure under substitutions, and formulate our main results.
%semilinear sets (corresponding to rational sets of star height one) 
%are not enough to capture Parikh images of \kNRA 1.
%In the next section we prove that Parikh images of \kNRA 1 are rational. 

\para{Orbit-finite unions}
Consider a family of sets $\X$.
We say that $\X$ is \emph{closed under orbit-finite unions} if for every orbit-finite
family $(X_i)_{i\in I}$ of sets $X_i \in \X$, 
the union
$
\bigcup_{i\in I} X_i
$
belongs to $\X$. 
We instantiate below this abstract definition to families $\X$ of
sets of data words and data vectors.

\para{Rational data languages}
%We start with a generalisation of regular expressions to orbit-finite alphabets.
%The lifting amounts to relaxing finite sums to orbit-finite ones.
%In the following we fix an orbit-finite alphabet $\Sigma$ and consider data languages $L\subseteq \Sigma^*$.
We consider data languages over a fixed orbit-finite alphabet $\Sigma$.

As usual, we define concatenation of two data languages
$
L L' = \setof{ w w'}{w\in L, w'\in L'},
$
and the Kleene star (iteration):
%the powers $L^n$, $n\geq 0$, inductively: $L^0 = \set{\varepsilon}$ and $L^{n+1} = L L^n$;
%and the iteration $L^*$ as the union of all powers:
$
L^* = \setof{w_1  \ldots  w_n}{n\geq 0, w_1, \ldots, w_n\in L}.
%L^* \ = \ \bigcup_{n\geq 0} L^n.
$
%For a fixed alphabet $\Sigma$, 
Let \emph{rational data languages} % over $\Sigma$ 
be the smallest class of data languages
that contains all singleton languages $\set{w}$, for $w\in\Sigma^*$,
and is closed under concatenation, iteration, and orbit-finite unions.
%The latter means that for every orbit-finite set $I$
%and a finitely supported indexed family $(L_i)_{i\in I}$ of rational languages $L_i\subseteq \Sigma$ 
%(formally, we mean here a finitely supported function $i \mapsto L_i$ from $I$ to 
%$\powfs{\Sigma^*}$), their union
%\[
%\bigcup_{i\in I} L_i
%\]
%is also rational.
In particular the empty language, all finite languages and all orbit-finite ones are rational.
%Note that by Lemma~\ref{lem:o-f}, orbit-finite unions preserve orbit-finiteness of an alphabet.
For finite  $\Sigma$ we obtain the classical rational (regular) sets.
As expected, without the Kleene star we obtain
% star-free data languages are 
exactly sets of words of bounded length, 
or equivalently, due to Lemma~\ref{lem:b-o-f}, orbit-finite languages.

When convenient, we may speak of \emph{rational expressions}, by which we mean formal derivations
of rational languages according to the closure rules listed above. 

\begin{example} \label{ex:cont}
Continuing Example~\ref{ex:L1L2}, the language $L_2$  is rational, as it can be presented by a rational expression:
\[
L_2 \ = \ \Big(\bigcup_{a,b \in\atoms, a\neq b} \set{ab} \Big)^* \, \big(\set{\emptyword} \, \cup \bigcup_{a\in\atoms} \set{a}\big).
\]
For readability we omit brackets $\set{}$ in the sequel, as in~\eqref{eq:L2rat}.
On the other hand, one easy shows that the language $L_1$ %from Example~\ref{ex:L1L2}
 is not rational.  %(cf.~Lemma~\ref{lem:L1} in Appendix~\ref{sec:rat-app}).
% 
% for sufficiently long words $w\in L$,
%
%
%We claim there exists at least one star expression $R'$ with the following property: 
%Every atom in the expression $R'$ is different from the support of the indexing set of the immediate outer orbit-finite union. If not, $L$ cannot be infinite. Choose a word $w \in L$ in which $R'$ is used at least once. Let $w = xuy$ where $u$ is the infix of $w$ corresponding to the repetition of $R'$. Since all atoms in $u$ are not in the support, we form a new word $w' = xuu'y$ by performing $R'$ with $u' = \theta.u$ where $\theta$ is a permutation that swaps first and last atom of $u$. But $w'$ does not belongs to $L_1$ because last atom of $u$ is equal to first atom of $u'$. Therefore, such a $R$ cannot exist.       
%%
\qed
\end{example}

\para{Rational sets of data vectors}
We consider sets of data vectors over a fixed orbit-finite alphabet $\Sigma$.
Let addition of two sets $X,Y$ of data vectors be defined by Minkowski sum
\[
X + Y = \setof{x+y}{x\in X, y\in Y},
\]
and let $X^*$ contain all finite sums of elements of $X$:
\[
X^* = \setof{x_1 + \ldots + x_n}{n\geq 0, x_1, \ldots, x_n\in X}.
\]
We define \emph{rational sets} of data vectors as the smallest class of 
sets of data vectors that contains all singletons $\set{x}$ and is closed under addition, star, and orbit-finite unions.
In particular, the empty set, all finite sets and all orbit-finite sets of data vectors are rational.
\begin{example} \label{ex:ParL1}
Continuing Example~\ref{ex:cont},
the Parikh image of $L_1$ (and $L_2$) is rational (for readability we keep omitting brackets $\set{}$):
%(cf.~\eqref{eq:L2rat} in Example~\ref{ex:L1L2}):
%
\begin{align*} %\label{eq:L2rat}
\Par{L_1} \ = \ \Big(\bigcup_{a,b \in\atoms, a\neq b} a + b \Big)^* \, + \, \big(\zerovector  \,\, \cup \bigcup_{a\in\atoms} a\big).
\end{align*}
\end{example}
%
%The following fact will be useful later:
%%
%\begin{claim} \label{claim:onlyif}
%Parikh image of a rational data language is a rational set of data vectors.
%\end{claim}
%
\begin{claim} \label{claim:iff}
(1) Rational sets of data vectors are exactly Parikh images of rational data languages.
(2) $\Par L$  is  rational if, and only if,
$L$ is Parikh-equivalent to a rational data language.
\end{claim}
%
%The converse of Claim~\ref{claim:onlyif} is not true, as witnessed by the language $L_1$ from Example~\ref{ex:L1}
%which itself is not rational (cf.~Claim~\ref{claim:L1}) but,
%due to Parikh-equivalence to a rational language $L_2$ (cf.~Examples~\ref{ex:L2} and~\ref{ex:L1L2}).
%its Parikh images is:
%%
%\begin{claim} \label{claim:L1-Par}
%Parikh image of the language $L_1$ is rational.
%\end{claim}
%

\begin{remark} \rm
The classical notion of rational sets in an arbitrary monoid (\cite[Chapter VII]{Eilenberg74})
can be generalised along the same lines as above to sets with atoms, by considering orbit-finite unions instead
of finite ones.
In this paper we stick to monoids of data words and data vectors, over an orbit-finite alphabet.
\end{remark}

\para{Substitutions}
Consider a language $L$ over an orbit-finite alphabet $\Sigma$ and
a (legal) family of languages $K = (K_\sigma)_{\sigma\in\Sigma}$ over an alphabet $\Gamma$,
indexed by $\Sigma$. We typically use the anonymous function notation
\[
\sigma \quad \mapsto \quad K_\sigma.
\]
The \emph{substitution} $L(K)$ is the language over $\Gamma$ containing all words obtained
from some word $\sigma_1 \sigma_2 \ldots \sigma_n \in L$, by replacing every letter $\sigma_i$
by some word from $K_{\sigma_i}$:
\[
L(K) \ = \ \bigcup_{\sigma_1 \sigma_2 \ldots \sigma_n \in L}
K_{\sigma_1} K_{\sigma_2} \ldots K_{\sigma_n}.
\]
\begin{example}
As usual we use the shorthand $L^+ = L^* L$.
Consider the language $L_1$ from Example~\ref{ex:L1L2} and
$\Sigma = \Gamma = \atoms$.
By the equivariant substitution $K_a = (aa)^+$,  or 
$
a \ \mapsto \ (a a)^+,
$
we obtain the language $L_1(K) = \big(\bigcup_{a\in\atoms} aa\big)^+$ containing words, where
all maximal constant infixes have even length. 
\end{example}
\begin{lemma} % [Appendix~\ref{sec:rat-app}] 
\label{lem:subst}
If $L$ and all languages $K_\sigma$ have rational Parikh images (resp.~are rational)
then the substitution $L(K)$ has also rational Parikh image (resp.~is rational).
\end{lemma}
\begin{proof} % [Proof of Lemma~\ref{lem:subst}]
Intuitively speaking, it is enough to replace syntactically, in the rational expression defining $\Par L$,
every appearance of a letter $\sigma$ by an expression defining $\Par {K_\sigma}$.

Formally, we proceed by induction on a derivation of $L$.
By Claim~\ref{claim:iff}(2) we assume, w.l.o.g., that languages $L$ and $K_\sigma$ are rational.
If $L = \set{\sigma}$ is a singleton, $\sigma\in\Sigma$, then $L(K) = K_\sigma$ and hence is rational.
The cases of $L = L_1  L_2$, or $L = (L')^*$,
are both immediate, as both the operations preserve rationality,
and $L_1, L_2$ and $L'$ are rational by induction assumption.
Finally, when $L = \bigcup_{x\in X} L_x$, 
%for an orbit-finite set $X$ and a finitely supported mapping $x\mapsto L_x$,
by induction assumption we know rationality of the languages $L_x(K)$ for $x\in X$.
As 
\[
L(K) \ = \ \bigcup_{x\in X} L_x(K) 
\]
and the mapping $x \mapsto L_x(K)$ is supported by the union of the supports of $x\mapsto L_x$
and $\sigma \mapsto K_\sigma$, we deduce that $L(K)$ is an orbit-finite union of rational sets and
hence rational.
\end{proof}

%\medskip

\para{Main results}
As our main contribution, we prove rationality of Parikh images of \kNRA 1 and \kCFG 1:

\begin{theorem} \label{thm:1ARA}
Parikh images of \kNRA 1 languages are rational.
\end{theorem}
\begin{theorem} \label{thm:1G}
Parikh images of  \kCFG 1 languages  are rational.
\end{theorem}
\noindent
%Consider a fixed \kNRA 1 $\A = \tuple{H, Q, I, F, \Delta}$.
We actually prove a refined version of Theorem~\ref{thm:1ARA}
(needed also for proving Theorem~\ref{thm:1G})
which, due to~\eqref{eq:L}, implies Theorem~\ref{thm:1ARA}:
%For a \kNRA 1 $\A$, let $\langoffull \A q a {q'} {a'}$ be the set of data words input by a run starting
%in a configuration $\conf q a$ and ending in a configuration $\conf {q'} {a'}$.
%
\begin{lemma} \label{lem:1ARA}
For every \kNRA 1 $\A$, the languages $\langoffull \A q a {q'} {a'}$ have rational Parikh images.
\end{lemma}
\noindent
%, as $\langof \A$ is an orbit-finite union:
%\[
%\langof \A \ = \ \bigcup_{q\in I, q'\in F, a,a'\in \atoms} \langoffull \A q a {q'} {a'}.
%\]

Before proving Theorem~\ref{thm:1G} and Lemma~\ref{lem:1ARA} in Sections~\ref{sec:1CFG}--\ref{sec:anti},
we first demonstrate that semi-linear sets are not sufficient to capture Parikh images 
of \kCFG 1, or even \kNRA 1.

\section{Semi-linear sets}   \label{sec:sl}

Analogously to rational sets, we lift semi-linear sets to orbit-finite alphabets.
Consider data vectors over a fixed orbit-finite alphabet $\Sigma$.
A \emph{linear set} is then any set of the form
\[
N \ = \ g + P^*
\]
for a data vector $g$ and an orbit-finite set $P$ of data vectors, and a semi-linear set is any orbit-finite union of linear sets:
\begin{align} \label{eq:sl}
\bigcup_{i\in I} \, N_i \quad = \quad 
\bigcup_{i\in I} \, g_i + {P_i}^*.
\end{align}
In particular, $I$ is orbit-finite, and the function $i \mapsto (g_i, P_i)$ mapping $i\in I$
to a  data vector $g_i$ (base) and an orbit-finite set $P_i$ of data vectors (periods) is legal.
By definition, semi-linear sets are a subset of rational sets of \emph{star-height 1}
(star-height is defined as usual, as the maximal nesting depth of stars).

\begin{example}
Parikh image of $L_1$  (cf.~Example~\ref{ex:ParL1}, $\Sigma=\atoms$) 
is semi-linear:   (with all sets $P_i$ equal): 
%Indeed, it is enuogh to take:
\[
I \ = \ \big(\zerovector \,\, \cup \bigcup_{a\in\atoms} \! a\big) \qquad
g_i \ = \ i \qquad
P_i  \ = \! \bigcup_{a,b \in\atoms, a\neq b} \!\!\! a + b.
% \, + \, \big(\emptyword \, \cup \bigcup_{a\in\atoms} a\big).
%\big(\bigcup_{a,b \in\atoms, a\neq b} a b \big)^* 
%\quad \cup \quad
%\big(\bigcup_{a,b \in\atoms, a\neq b} a b \big)^* \, + \, \bigcup_{a\in\atoms} a
\]
\end{example}

\begin{proposition}%[Appendix~\ref{sec:sl-app}]  
\label{prop:sh1}
Semi-linear sets of data vectors are exactly rational sets of star-height at most $1$.
\end{proposition}

\para{Semi-linear sets are not sufficient}
We demonstrate that Parikh images of \kNRA 1 languages are not semi-linear in general.
As a counterexample we take the following language $L_3\subseteq \atoms$.
For $a\in\atoms$, let
\begin{align} \label{eq:Marta}
K_{a} \ = \ \bigcup_{b \in\atoms-\set{a}} b; \qquad 
L_a \ = \ a a \, \big(a\, K_{a} \big)^*.
\end{align}
Let $L_3$ be the language obtained from $L_1$ by the substitution:
$
a \ \mapsto \ L_a.
$
%\[
%\xymatrix{
%{}\ar[r]&
%*+[o][F-]{p}\ar@/^/[rr]^(.2){\top}^(.8){\downarrow 1}\ar@(dr,dl)[]^(.2){\top} 
%& & *+[o][F-]{q}\ar@/^/[rr]^(.3){x=x_1}\ar@(dr,dl)[]^(.3){x\neq x_1}
%& & *+<9pt>[o][F=]{r}\ar@(ur,dr)[]^(.2){x=x_1}\ar@/^/[ll]^(.2){x\neq x_1}
%}
%\]
The language is clearly rational, and recognised by a (deterministic) one-register automaton ($H$ is omitted):
\vspace{-1mm}
\begin{figure}[H] %[btp]
\begin{center}
\includegraphics[width=7cm]{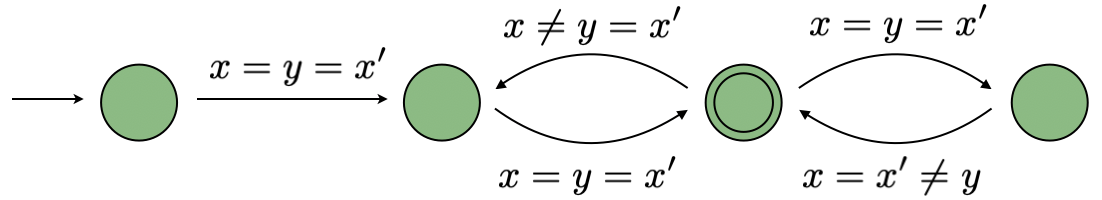}
\end{center}
\end{figure}
\vspace{-3mm}
As we show, its Parikh image is not semi-linear, which
motivates consideration of rational sets
%, instead of only semi-linear ones,
in forthcoming sections.

\begin{lemma}%[Appendix~\ref{sec:sl-app}]  
\label{lem:Marta}
$\Par{L_3}$ % \kNRA 1 languages are
is not semi-linear.
\end{lemma}
\noindent
Let  $\poj v = \size{\setof{a\in \atoms}{v(a)=1}}$
denote the number of atoms appearing exactly once in $v$.
The argument relies on a careful analysis of the limit value of the \emph{singularity ratio} 
$\frac {\poj v}{\size v}$ for $v \in \Par{L_3}$, 
when $\size v$ tends to infinity.

% !TEX root = main.tex
 
\section{Proof of Theorem~\ref{thm:1G}}  \label{sec:1CFG}

%
%
%The rest of this section is devoted to proving our last result:
%%
%\begin{theorem} \label{thm:1G}
%Parikh images of  \kCFG 1 languages  are rational.
%\end{theorem}
%

Consider a fixed \kCFG 1 $\G = (H, Q, q_0, \Delta_2, \Delta_0)$.
%
%\para{Normal form}
%

%Moreover we assume, w.l.o.g., that no nonterminal $q\in Q$ has a production rule from $\Delta_2$
%as well as one from $\Delta_0$. Thus $Q$ splits into two disjoint subsets $Q = Q_2 \cup Q_0$.

%\para{Notations for trees}

\para{Proof strategy}

We proceed in three steps.
First, by a Ramsey's argument,  we prove that a sufficiently large set of productions
contains a \emph{compatible} pair (Lemma~\ref{lem:Ramsey}).
Then we define \emph{width} of derivation trees and show that for a sufficiently large $n\in\N$, 
every derivation tree can be transformed into a tree of width at most $n$
while preserving the Parikh image of its yield (Lemma~\ref{lem:tobw}).
The cut-and-paste transformation relies on compatibility of productions in a tree. 
Finally,
%, relying on Lemma~\ref{lem:1ARA}, 
we argue that Parikh image of the set of words generated by derivation trees of 
width bounded by $n$ is rational, for every fixed $n\in\N$ (Lemma~\ref{lem:w}).
Lemmas~\ref{lem:tobw} and~\ref{lem:w} imply Theorem~\ref{thm:1G}.

\para{Compatibility}

The equality type of a tuple $\tuple{a_1, \ldots, a_k}\in \atoms^k$ is defined as
the set $\setof{\pair i j}{1 \leq i < j \leq k, \ a_i = a_j}$.
Intuitively speaking, tuples of the same equality type admit the same equalities between their coordinates.
Two tuples $\alpha = \tuple{a_1, \ldots, a_k}$ and $\beta = \tuple{b_1, \ldots, b_k}$ 
we call \emph{compatible} if they have the same equality type, and
for every coordinate $i\in\set{1, \ldots, k}$ one of two conditions holds: either
(1) $a_i = b_i$; or (2) $a_i \neq b_i$ and both $a_i$ and $b_i$ do not appear in the other tuple:
% elsewhere in $\alpha$ and $\beta$:
$a_i \notin\set{b_1,\ldots, b_k}$, $b_i\notin\set{a_1,\ldots,a_k}$.
%$\set{a_i, b_i} \cap \set{a_j, b_j} = \emptyset$ for every $j\neq i$.
In particular, two equal $k$-tuples are always compatible.

\begin{lemma} \label{lem:Ramsey}
For every $k\in\N$ there is some $l = f(k)\in\N$
such that every finite multiset of $k$-tuples of atoms $A : \atoms^k \to \N$ of size at least $l$
contains two compatible $k$-tuples.
\end{lemma}
\begin{proof}
Let $k\in\N$ be fixed.
If $A$ contains two equal tuples, they are compatible. 
Thus we can assume $A$ to be a set.
We take $l = f(k)$ large enough to satisfy the constraint~\eqref{eq:constr} below.

The number of different equality types $E_k$ is finite and equal to the number of partitions
of the coordinates set $\set{1, \ldots, k}$ (the $k$th Bell number).
By the pigeonhole principle, for $l = \size A$ large enough, there is a subset $A' \subseteq A$
of size $l' = \size{A'} = \frac l {E_k}$ whose elements have all the same equality type. 

We now consider an undirected clique of size $l'$ with vertices $A'$, where the 
edge between vertices $\alpha = \tuple{a_1, \ldots, a_k}$ and $\beta = \tuple{b_1, \ldots, b_k}$ 
is labeled (coloured) by the set $D_{\alpha \beta} = \setof{i\in\set{1,\ldots, k}}{a_i \neq b_i}$.
Intuitively, the colour describes the coordinates on which $\alpha$ and $\beta$ disagree.
The number of colours is at most $C = 2^k$.
By Ramsey's theorem, for $l'$ large enough
the graph contains a monochromatic clique $A''$ of size
$l'' = k^2+1$; indeed, it suffices to take
\begin{align} \label{eq:constr}
l' \ \geq \ R(\underbrace{l'', l'', \ldots, l''}_C).
\end{align}
Thus every two elements of $A''$ disagree on the same coordinates $D\subseteq \set{1,\ldots, k}$,
and hence also agree on the same coordinates $\set{1,\ldots, k} - D$.

Take any $\alpha = \tuple{a_1, \ldots, a_k} \in A''$.
For every coordinate $i\in D$, all tuples $\beta \in A''$ are pairwise different on that coordinate.
Therefore, at most $k$ tuples $\beta = \tuple{b_1, \ldots, b_k} \in A''$ may satisfy 
\begin{align} \label{eq:bi}
b_i \in \set{a_1, \ldots, a_k},
\end{align} 
i.e., $b_i$ appears in $\alpha$.
As $\size D\leq k$, at most $k^2$ tuples (including $\alpha$ itself) may satisfy
the condition~\eqref{eq:bi} for some coordinate $i\in D$.
Therefore taking any of the remaining tuples, say $\beta$, we obtain a compatible pair $\alpha, \beta$.
\end{proof}

\para{Traversals and side-effects}

The number of children of a node $x$ in a derivation tree $\T$ we call \emph{arity} of $x$
(leaves are nodes of arity 0).
Let $\preceq$ denote the tree order ($x\preceq y$ if $x$ is an ancestor of $y$).
A path from a node $x$ to a node $y$, assuming $x\preceq y$, is the set 
$\setof{z\in \T}{x\preceq z \preceq y}$ of all nodes $z$ appearing between
the nodes $x$ and $y$, including $x$ and $y$.

Consider an arbitrary derivation tree $\T$ of $\G$.
We distinguish two ways of traversing a production 
${\conf q {a}} \trans{} {\conf {p} {b}} \,\, {\conf {p'}{b'}} \in \Pi_2$ appearing in $\T$ by a path,
namely left and right traversal:
%A path in a derivation tree contains either left or right \emph{side} of every traversed production from $\Delta_2$:

\vspace{-2mm}
\begin{figure}[H] %[btp]
\begin{center}
\includegraphics[width=5.3cm]{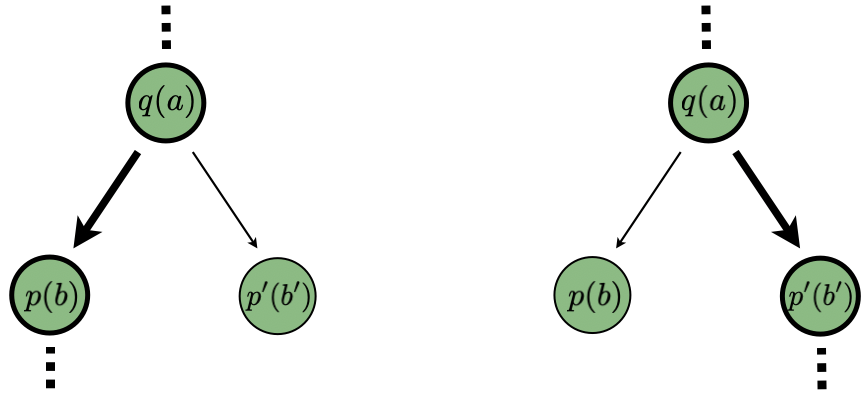}
\end{center}
\end{figure}
\vspace{-3mm}

\noindent
%Clearly, a production appearing in a derivation tree par
%
Once left or right traversal is chosen, say the right one,
a production ${\conf q {a}} \trans{} {\conf {p} {b}} \,\, {\conf {p'}{b'}} \in \Pi_2$
resembles a transition of \kNRA 1 (over the extended input alphabet 
$\Gamma=(Q\cup H)\times\atoms$) from $\conf q a$ to $\conf {p'} {b'}$ 
which inputs the label of the remaining node, namely $\conf p b$.
We call the pair $\conf p b\in \Gamma$ the \emph{side-effect} of the right traversal; 
symmetrically we call $\conf {p'} {b'}$ the side-effect of the left traversal.
% The sequence of side-effects along a path we call briefly the side-effect of the path.
For two configurations $\conf q a$ and $\conf p b$  of $\G$, we denote
by $S_{\conf q a\, \conf p b} \subseteq \Gamma^*$ 
the set of all sequences of side-effects that may appear along a path from a node labeled by $\conf q a$
to a node labeled by $\conf p b$ in a derivation tree of $\G$.
As a corollary of Lemma~\ref{lem:1ARA} we get:
\begin{lemma} \label{lem:p}
Languages $S_{\conf q a\, \conf p b}$ have rational Parikh images.
\end{lemma}
\begin{proof}
Indeed, the claim follows immediately by Lemma~\ref{lem:1ARA}, if production traversals are considered
as transitions of a \kNRA 1 over the input alphabet $\Gamma$, and the side-effect of a traversal
is considered as input of a transition.
\end{proof}

\para{Height, width, and rank}

Recall the normal form of constraints \ee--\nnn as defined in Section~\ref{sec:pre}.
Similarly as in case of \kNRA 1, the right traversal of a production 
${\conf q {a}} \trans{} {\conf {p} {b}} \,\, {\conf {p'}{b'}} \in \Pi_2$
is called \emph{register-preserving} if $a = b'$, and \emph{register-updating} if
$a \neq b'$; likewise for the left traversal.
%Analogously, a production ${\conf q {a}} \trans{} {\pair h b} \in \Pi_0$ is called register-updating if $a\neq b$,
%and register-preserving otherwise.

We define the \emph{length} of a path in a derivation tree $\T$ as the number of register-updating production traversals along
the path,
and the \emph{height} of a node $x$ in $\T$ as the maximal length of a path from $x$ to a leaf.
%The height of a tree is the height of its root.
%
A \emph{cut} in $\T$ is a set of nodes which are pairwise incomparable with respect to the tree ordering.
A cut is called \emph{$n$-cut} if its size is at least $n$ and the height of every node in the cut is at least $n$.
The \emph{width} of a derivation tree $\T$ is the maximal $n$ for which  $\T$ contains some $n$-cut.

%\subsection{Transforming derivation trees to bounded width}
\label{sec:tobw}

The \emph{rank} of a derivation tree is defined as the multiset of lengths of all paths from the root to some leaf.
For a finite multiset $r : \N \to \N$ of natural numbers, let the \emph{diagram of} $r$ be the unique
non-increasing sequence $w\in \N^*$ such that $\Par w = r$.
We define the order on ranks as follows: $r \leq r'$ if the diagram of $r$ is lexicographically smaller 
than the diagram of $r'$.
For instance, $\set{7, 5, 2, 2} < \set{7, 7, 3}$.

We call two derivation trees $\T, \T'$ Parikh-equivalent if
$\Par {\yield \T} = \Par {\yield {\T'}}$.

\begin{lemma} \label{lem:tobw}
For a sufficiently large $n$, 
every derivation tree is Parikh-equivalent to a derivation tree of width at most $n$.
\end{lemma}

\begin{proof}
Let $m = \size{\Delta_2}$.
Fix an arbitrary $n\geq f(6) \cdot 2m$, for $f$ given by Lemma~\ref{lem:Ramsey}.  % be fixed (as revealed below).
We show:
\begin{claim} \label{claim:transf}
Every derivation tree $\T$ of $\G$ of width $\geq n$ can be transformed,
by cutting and pasting of some parts, into a Parikh-equivalent derivation tree $\T'$ of 
rank strictly larger than $\T$, but of the same size (= the number of nodes) as $\T$.
\end{claim}
\noindent
The claim is sufficient for proving Lemma~\ref{lem:tobw}.
Indeed, as the transformation preserves the size, 
the rank can increase only finitely many times.
Therefore, by iterating the transformation 
we ultimately arrive at a derivation tree $\T'$ whose rank can not be further increased. 
By Claim~\ref{claim:transf}, the width of $\T'$ is forcedly at most $n-1$, as required. 

From now on we concentrate on proving Claim~\ref{claim:transf}.
Let $\T$ be a derivation tree of width $\geq n$.
Consider some fixed $n$-cut $\set{x_1, \ldots, x_n}$ and disjoint paths $\pi_1, \ldots, \pi_n$ in $\T$
of length $\geq n$, 
each path $\pi_i$ going from $x_i$ to some leaf. 

Consider a fixed path $\pi_i$.
It contains $\geq n$ register-updating production traversals, and therefore
by the pigeonhole principle
the same production rule $q \trans{\varphi} p\,\, p' \ \in \ \Delta_2$ and the same (say left) 
register-updating traversal repeats at least $n' = \frac n {2m}$ times along $\pi_i$.
We apply Lemma~\ref{lem:Ramsey} for $k=3$ to deduce that, as $n'\geq f(6) \geq f(3)$, 
some two of these traversals

\vspace{-2mm}
\begin{figure}[H] %[btp]
\begin{center}
\includegraphics[width=5.3cm]{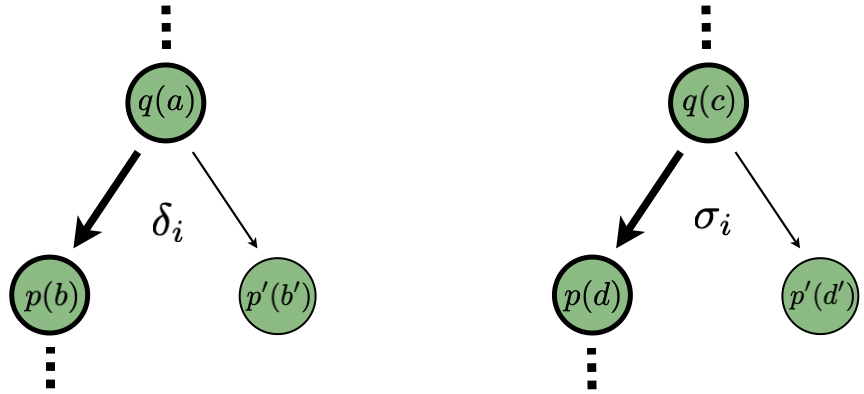}
\end{center}
\end{figure}
\vspace{-3mm}

\noindent
are compatible, by which we mean that 
their underlying $3$-tuples $\tuple{a,b,b'}$ and $\tuple{c,d,d'}$ are so.
Thus each path $\pi_i$ traverses a pair of compatible productions $\delta_i$, $\sigma_i$ 
which agree on the production rule and (left or right) traversal.

We now repeat a similar argument for paths.
As before, by the pigeonhole principle in at least $n'$ paths $\pi_i$, the same production rule and the same
traversal was used in productions $\delta_i$ and $\sigma_i$ derived in the above reasoning.
We now apply Lemma~\ref{lem:Ramsey} for $k=6$ to deal with pairs $\pair {\delta_i} {\sigma_i}$ of productions,
where a pair $\pair {\delta_i} {\sigma_i}$ induces a 6-tuple obtained 
by concatenating two underlying 3-tuples of $\delta_i$ and $\sigma_i$.
Since $n'\geq f(6)$, according to the lemma some two of these pairs,
say $\pair {\delta_i} {\sigma_i}$ and $\pair {\delta_j} {\sigma_j}$, are compatible
(by which we mean that the two induced $6$-tuples are so).

We have thus four productions $\delta, \sigma, \bar \delta, \bar \sigma$, 
traversed  by two disjoint paths in $\T$
(we do not depict nonterminals as all the four productions are induced by the same rule):

\vspace{-3mm}
\begin{figure}[H] %[btp]
\begin{center}
\includegraphics[width=6.5cm]{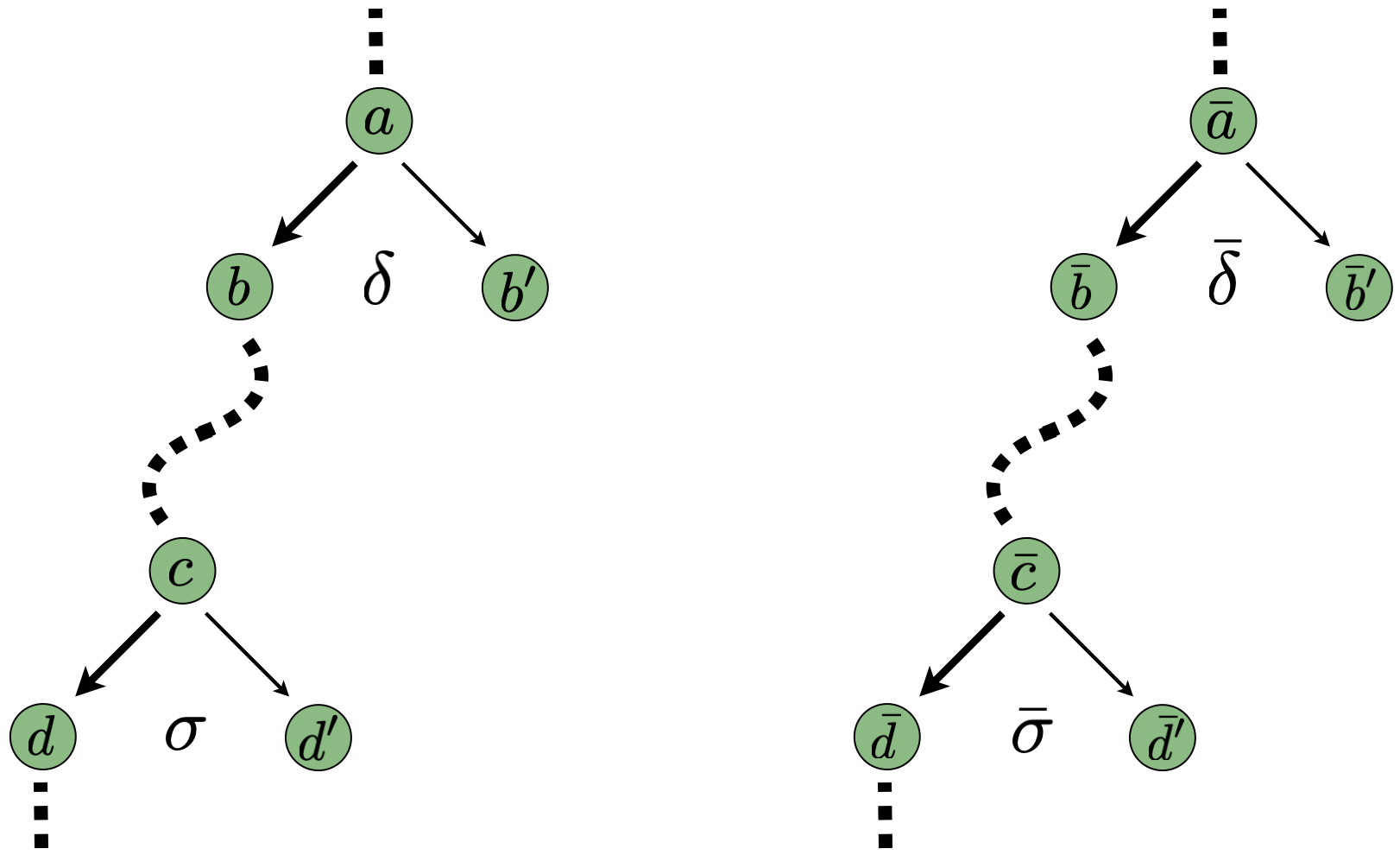}
\end{center}
\end{figure}
\vspace{-5mm}

%\noindent
%such that their underlying tuples 
%$\tuple{a,b,b'}$, $\tuple{c,d,d'}$,
%$\tuple{\bar a, \bar b, \bar b'}$ and $\tuple{\bar c, \bar d, \bar d'}$ 
%have the same equality type and satisfy:  % the following property:

\begin{claim} \label{claim:4}
The four underlying triples 
$\tuple{a,b,b'}$, $\tuple{c,d,d'}$,
$\tuple{\bar a, \bar b, \bar b'}$ and $\tuple{\bar c, \bar d, \bar d'}$ 
%have the same equality type and 
are pairwise compatible.
%
%For every $i\in\set{1, 2, 3}$, either all four 3-tuples agree on coordinate $i$, 
%or each of the four atoms appearing on coordinate $i$ is different from all other atoms appearing elsewhere in the four 
%3-tuples.
\end{claim}
\begin{proof}[Proof of Claim~\ref{claim:4}]
By the construction we have compatibility of triples $\tuple{a,b,b'}$ and $\tuple{c,d,d'}$,
and of triples $\tuple{\bar a, \bar b, \bar b'}$ and $\tuple{\bar c, \bar d, \bar d'}$.
Furthermore, we have also compatibility of
6-tuples $\tuple{a,b,b',c,d,d'}$ and $\tuple{\bar a, \bar b, \bar b', \bar c, \bar d, \bar d'}$, which implies
compatibility of triples $\tuple{a,b,b'}$ and $\tuple{\bar a, \bar b, \bar b'}$,
and of $\tuple{c,d,d'}$ and $\tuple{\bar c, \bar d, \bar d'}$.
Therefore, it only remains to prove compatibility of  $\tuple{a,b,b'}$ and $\tuple{\bar c, \bar d, \bar d'}$, and
of  $\tuple{c,d,d'}$ and $\tuple{\bar a, \bar b, \bar b'}$.
We concentrate of the former pair, as the other one is dealt with similarly.

The equality types of triples $\tuple{a,b,b'}$ and $\tuple{\bar c, \bar d, \bar d'}$ are the same, since so are
the equality types of $\tuple{a,b,b'}$ and $\tuple{c,d,d'}$, and of 
$\tuple{c,d,d'}$ and $\tuple{\bar c, \bar d, \bar d'}$.
We concern the first coordinate of the triples.
Supposing $a \neq \bar c$, we derive $a \notin \set{\bar c, \bar d, \bar d'}$:
if $a \neq \bar a$ then this follows due to compatibility of the two 6-tuples, and 
if $a = \bar a$ then this follows due to compatibility of $\tuple{\bar a, \bar b, \bar b'}$ and $\tuple{\bar c, \bar d, \bar d'}$;
symmetrically we derive $\bar c \notin \set{a, b, b'}$.
The two remaining coordinates are dealt with similarly.
\end{proof}

We are now prepared to cutting and pasting in $\T$.
For convenience we use below atoms $a, b$, etc. to identify respective nodes 
(keeping in mind potential equalities between these atoms).
Recall that all the four traversals are register-updating, and hence $a\neq b$, and likewise for other tuples.
We distinguish three cases, depending on the relation of $b'$ to $a$ and $b$:

\para{Case 1 $b' = a$}
Define the \emph{relevance} of a node $x$ in $\T$ as the maximal length of a path from the root of $\T$ 
to a leaf that traverses $x$.
By symmetry assume, w.l.o.g., that the relevance $\bar r$ of the node $\bar b$ is larger or equal to 
the relevance $r$ of the node $b$.
We cut the segment of $\T$ starting from the edge $a \trans{} b$ and ending with the edge $c\trans{} d$, 
and paste this segment between the nodes $\bar a$ and $\bar b$
as depicted in the figure:

\vspace{-3mm}
\begin{figure}[H] %[btp]
\begin{center}
\includegraphics[width=6.5cm]{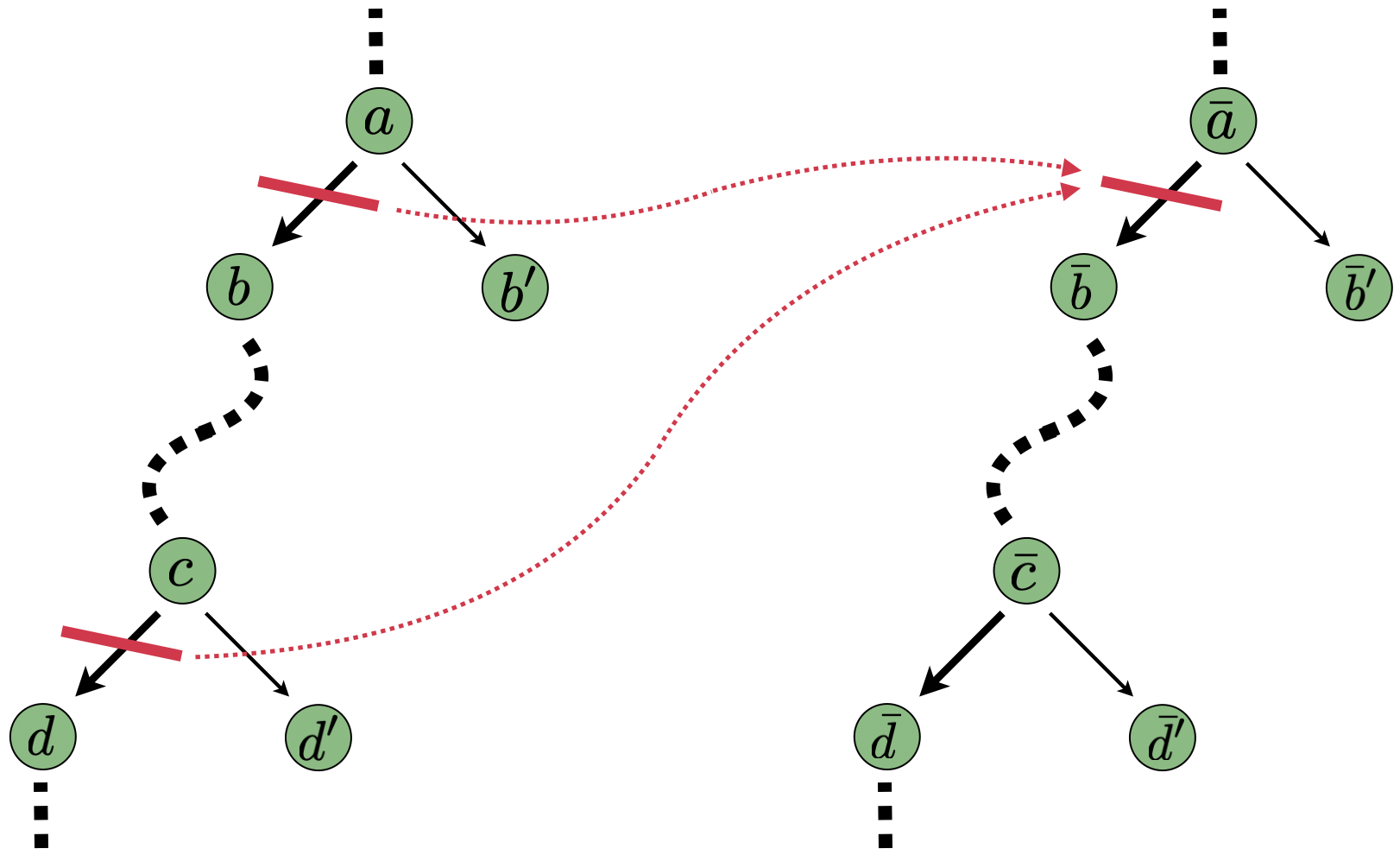}
\end{center}
\end{figure}
\vspace{-5mm}

\noindent
By Claim~\ref{claim:4} the tree $\T'$ so obtained is still a derivation tree:

\vspace{-3mm}
\begin{figure}[H] %[btp]
\begin{center}
\includegraphics[width=6.5cm]{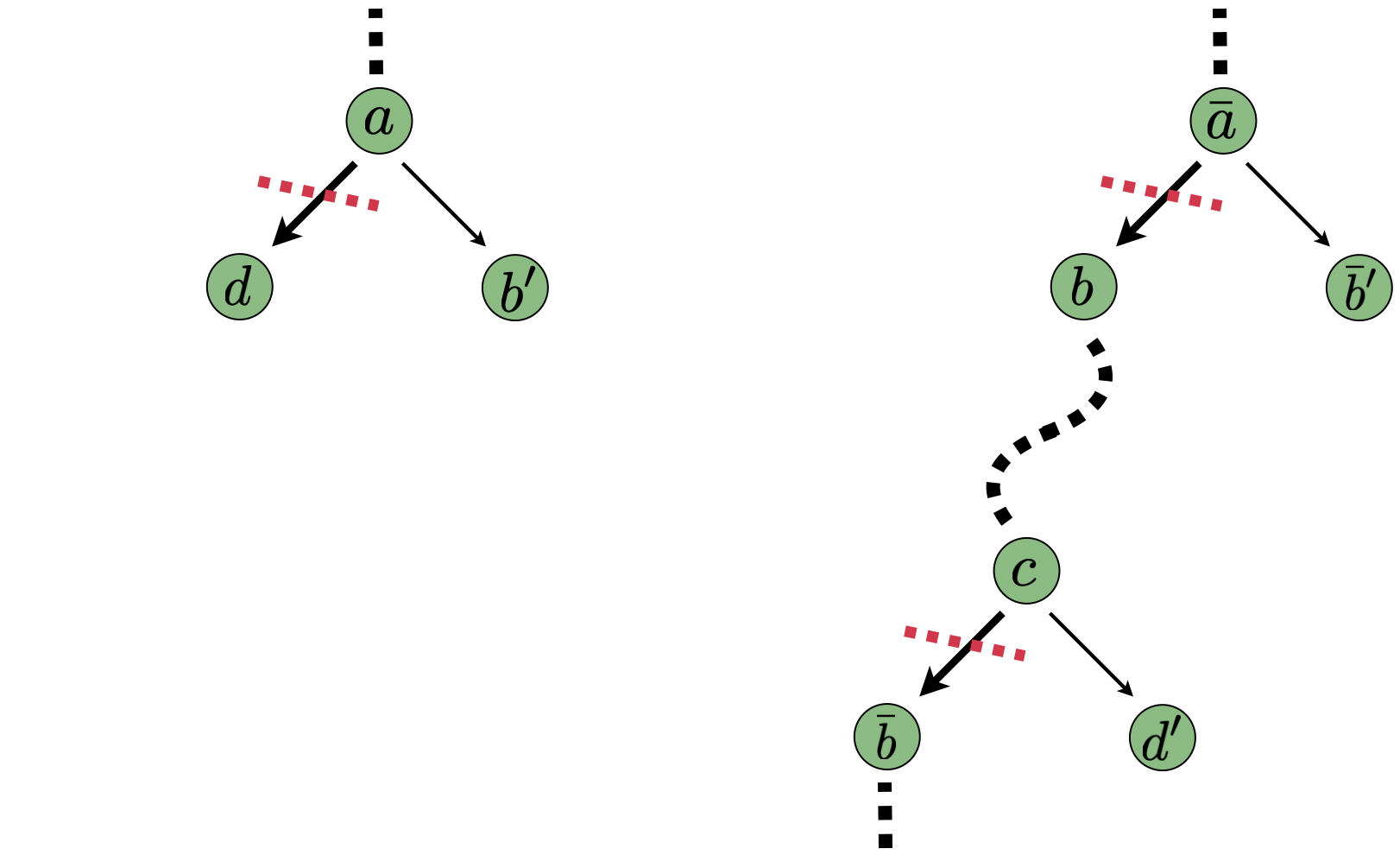}
\end{center}
\end{figure}
\vspace{-5mm}

\noindent
Indeed, $d\neq a$ (because either $d = b$ or $d$ does not appear elsewhere) and hence 
${\conf q {a}} \trans{} {\conf {p} {d}} \,\, {\conf {p'}{b'}} \in \Pi_2$ is a production;
likewise for the two remaining productions above.

Furthermore, we claim that rank of $\T'$ is strictly larger than rank of $\T$.
To this aim we analyse the effect of cut and paste on the lengths of the paths from the root to a leaf in $\T$.
First, all paths not traversing $b$ or $\bar b$ remain untouched.
%Concerning the cut, the lengths of all paths from the root to a leaf in $\T$ traversing
%the nodes $b$ and $d$ decrease, and the lengths of paths traversing $b$ but not $d$ change arbitrarily.
Furthermore, the lengths of all paths traversing $\bar b$ strictly increase.
Thus some path of length $\bar r$ in $\T$ gets strictly prolonged, and all other affected paths in $\T$
have lengths at most $r \leq \bar r$.
These two properties ensure that  the rank of $\T'$ is strictly larger than the rank of $\T$.

\para{Case 2 $b' = b$}

By symmetry assume, w.l.o.g., that the relevance $\bar r$ of the node $\bar a$ is larger or equal to 
the relevance $r$ of the node $a$.
We cut the segment of $\T$ starting from the edge $a \trans{} b$ and ending with edges $c\trans{} d$
and $c\trans{}d'$,
and paste this segment between the node $\bar a$ and the nodes $\bar b, \bar b'$,
and moreover cut the subtree rooted in $b'$ and paste it in place of the subtree rooted in $\bar b'$,
as depicted in the figure:

\vspace{-3mm}
\begin{figure}[H] %[btp]
\begin{center}
\includegraphics[width=6.5cm]{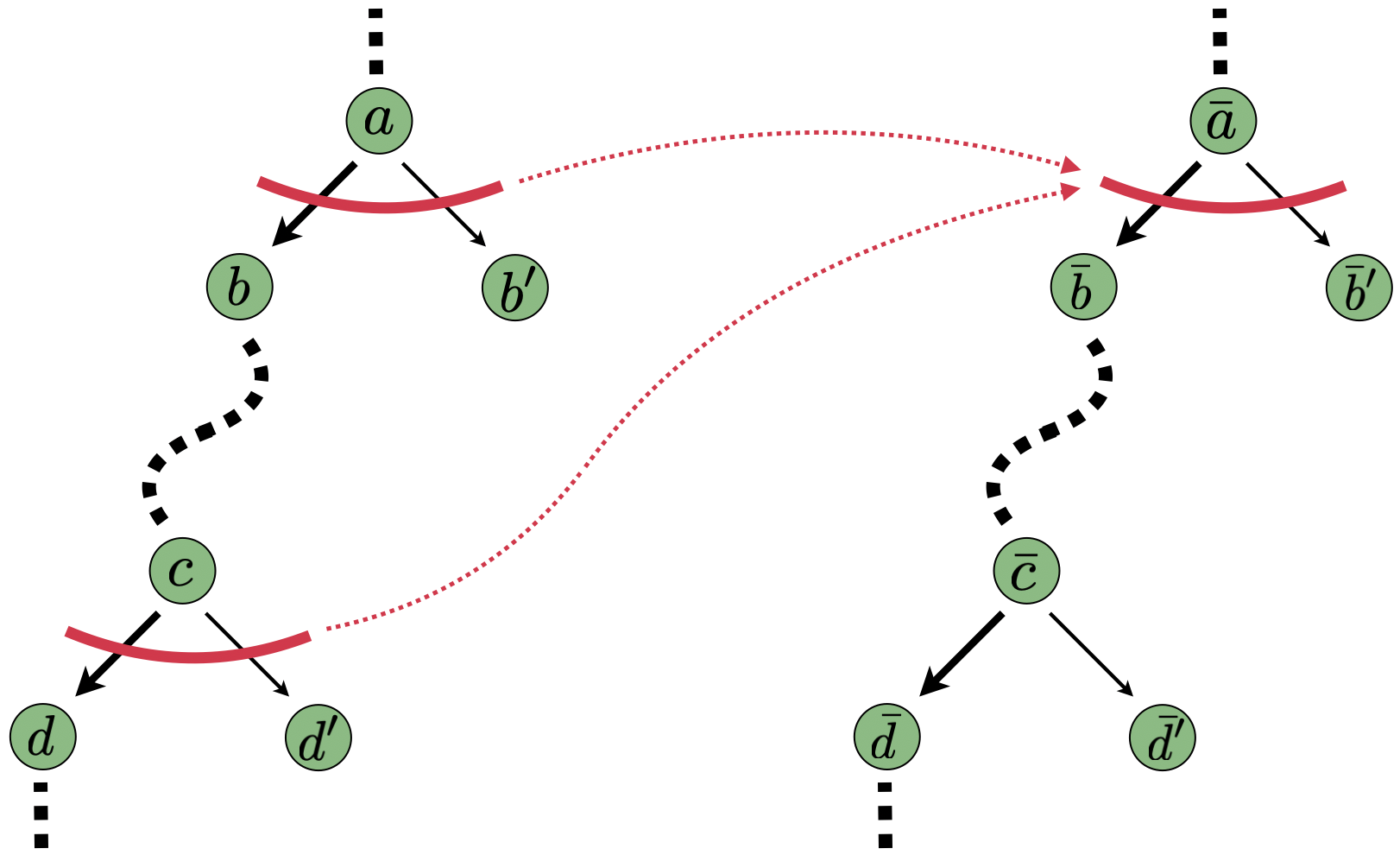}
\end{center}
\end{figure}
\vspace{-5mm}

\noindent
By Claim~\ref{claim:4} the tree $\T'$ obtained is a derivation tree, as before:

\vspace{-3mm}
\begin{figure}[H] %[btp]
\begin{center}
\includegraphics[width=6.5cm]{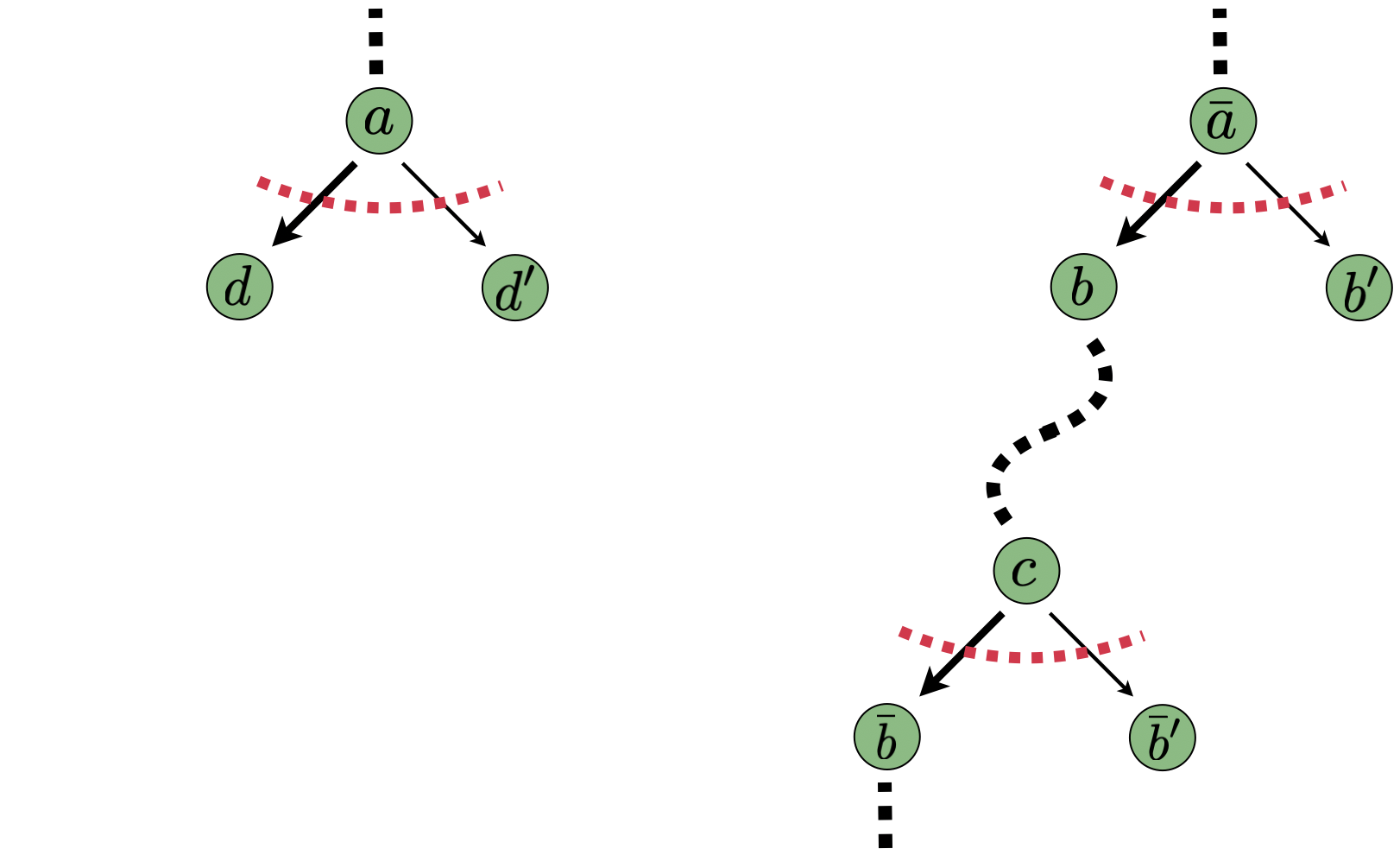}
\end{center}
\end{figure}
\vspace{-5mm}

\noindent
Similarly as before, we claim that the rank of $\T'$ is strictly larger than the rank of $\T$.
First, all paths from the root to a leaf in $\T$ not traversing $a$ or $\bar a$ remain untouched.
Furthermore,
the lengths of all paths from the root to a leaf that traverse $\bar a$ strictly increase.
%and the lengths of all paths traversing $a$ either decrease or change arbitrarily.
Thus some path of length $\bar r$ in $\T$ gets strictly prolonged, and all other affected paths in $\T$
have lengths at most $r \leq \bar r$.
As before, these two properties ensure that the rank of $\T'$ is strictly larger than the rank of $\T$.

\para{Case 3 $b' \notin \set{a,b}$}
In this case one can use any of the two cut-and-paste schemes described above.

\medskip

The proof of Claim~\ref{claim:transf} is thus completed, and hence so is the proof of Lemma~\ref{lem:tobw}.
\end{proof}

%\subsection{Rationality for bounded width}
%\label{sec:bw}

%
%We denote by $\langw n$ the subset of $\lang \G$ of words generated by a derivation tree of width at most $n$. 
%
We denote by $\langh {\conf q a,  n}\subseteq \langpar {\conf q a} \G$ the subset of words generated
by a derivation tree of height at most $n$, and by $\langw n$ the subset of $\lang \G$ of words
generated by a derivation tree of width at most $n$. 
We now prove, for every $n\in\N$, rationality of the languages $\langh {\conf q a,  n}$, 
and then use it to derive rationality of the Parikh image of the language $\langw n$.

\begin{lemma}%[Appendix~\ref{sec:1CFG-app}] 
\label{lem:h}
For every $n\in\N$, the languages $\langh {\conf q a, n}$ have rational Parikh images.
\end{lemma}

\begin{lemma} \label{lem:w}
For every $n\in\N$, the language $\langw n$ has rational Parikh image.
\end{lemma}
\begin{proof}%[Proof of Lemma~\ref{lem:w}]
For a fixed $n\in\N$,
consider an arbitrary derivation tree $\T$ of width at most $n$, and the subset $\H\subseteq \T$ of those nodes 
which have height at least $n+1$.
The set $\H$ is closed under ancestors and is thus itself a tree; contrarily to $\T$ whose non-leaf nodes
have arity 2, the tree $\H$ may contain nodes of arity 1.
Notably, as a special case $\H$ may be empty.

By assumption, width of $\T$ is at most $n$, and hence it may contain 
$n$-cuts but no $(n+1)$-cuts. This implies that the largest cut in $\H$ has size $n$.
In consequence:

\begin{claim} \label{claim:H}
$\H$ has at most $n$ leaves, and hence at most $n-1$ nodes of arity 2.
\end{claim}
\noindent
Let $\cal L$ denote the finite multiset (of size at most $n$) of configurations $\conf q a$ labelling leaves of $\H$.

\medskip
\noindent
\begin{minipage}{0.59\linewidth}
\indent
%All leaves of $\H$, the root of $\T$, all nodes of $\H$ of arity 2  and their children
%we call \emph{special} nodes of $\H$. 
%% There are thus at most $4n-2$ special nodes.
%Any path in $\H$ from a special node to another special node containing no other special nodes
%except the ends we call a \emph{segment}.
%
Any maximal path consisting of nodes of arity 1 we call a \emph{segment}.
Thus $\H$ decomposes uniquely into leaves, nodes of arity 2, and segments.
An example tree on the right has $n=4$ leaves, 3 nodes of arity 2 and 4 segments (depicted by blue areas)
of size 3, 2, 2 and 1, respectively. 
Using Claim~\ref{claim:H} we deduce:
\end{minipage}
\begin{minipage}{0.4\linewidth}
\vspace{-3mm}
\begin{figure}[H] %[btp]
\begin{center}
\includegraphics[width=3.0cm]{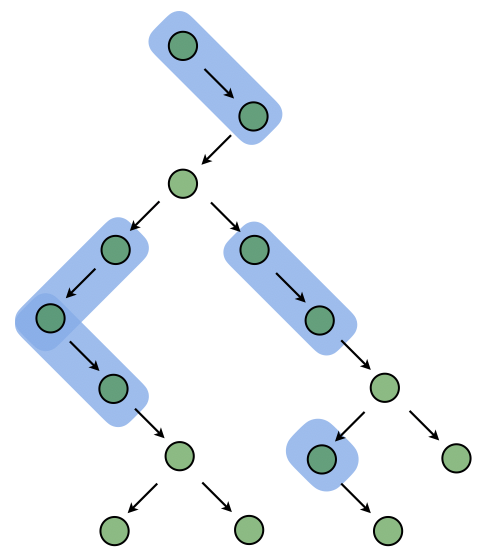}
\end{center}
\end{figure}
%\vspace{-5mm}
\end{minipage}

%\medskip
%\noindent
%
\begin{claim} \label{claim:S}
$\H$ contains at most $2n-1$ segments.
\end{claim}

Let $\cal S$ denote the finite multiset (of size at most $2n-1$) of pairs of configurations 
$\pair {\conf q a} {\conf p b}$ labelling ends of segments.
Let $\widetilde S_{\conf q a\, \conf p b}$ be obtained from the side-effect language 
$S_{\conf q a\, \conf p b}$ by the equivariant substitution (for $q' \in Q$):
\[
\conf {q'} c \quad \mapsto \quad \langh {\conf {q'} c, n-1};
\]
by Lemmas~\ref{lem:h} and~\ref{lem:p}  languages $\widetilde S_{\conf q a\, \conf p b}$ have thus rational Parikh images.
Let's define ($\prod$ denotes concatenation)
\[
L_{\cal L, \cal S} \quad = \quad 
\Big(\prod_{\conf q a \in \cal L} \langh {\conf q a, n}\Big)
%\quad + 
\quad 
\Big(\prod_{\pair {\conf q a} {\conf p b} \in \cal S} \widetilde S_{\conf q a\, \conf p b}\Big)
\]
as the concatenation of two concatenations, one of them ranging over $\cal L$ and the other one over $\cal S$.
By the very definition of the language 
$L_{\cal L, \cal S}$ we have
\begin{claim} \label{claim:incl}
$\Par {\yield \T} \in \Par {L_{\cal L, \cal S}}$.
\end{claim}
%
%Lemma~\ref{lem:w} is proved, once we show:
%
\begin{claim} \label{claim:Peq}
The languages $\langw n$ and
$K \ = \ \bigcup_{\cal L, \cal S} L_{\cal L, \cal S}$
are Parikh-equivalent,
where $\cal L, \cal S$ range over all possible sets arising from all derivation trees $\T$ of $\G$ of width at most $n$.
\end{claim}
\begin{proof} % [Proof of Claim~\ref{claim:Peq}]
The inclusion $\Par {\langw n} \subseteq \Par K$
we deduce by Claim~\ref{claim:incl}.
For the converse inclusion $\Par K \subseteq \Par {\langw n}$ we should prove: for every $\cal L, \cal S$ arising from some
derivation tree $\T$ of width at most $n$, the language $L_{\cal L, \cal S}$ is included in $\langw n$.
Indeed, given $\T$ and $\H$ used to derive sets $\cal L, \cal S$, we observe that every word $w \in L_{\cal L, \cal S}$
is Parikh-equivalent to 
the yield of a derivation tree $\T'$ of width at most $n$, 
obtained from $\H$ by replacing each leaf labelled by $\conf q a$ with a tree of height
$\leq n$ with root labeled by $\conf q a$, 
and replacing each segment with a sequence of productions, where every side-effect $\conf q a$ is replaced
by a tree of height at most $n-1$ with root labeled by $\conf q a$. 
Thus $\Par w\in \Par{\langw n}$.
\end{proof}

Finally, we derive rationality of $\Par K$.
By Lemmas~\ref{lem:h} and~\ref{lem:p} the languages $L_{\cal L, \cal S}$ have rational Parikh images.
Due to the bounds % of $n$ and $2n-1$, respectively, 
on
the size of $\cal L$ and $\cal S$ (cf.~Claims~\ref{claim:H} and~\ref{claim:S}), 
by Lemma~\ref{lem:b-o-f}
the set of all possible pairs $\cal L, \cal S$ is orbit-finite. 
Therefore $K$, as an orbit-finite union of languages with rational Parikh images, has a rational Parikh image too.
\end{proof}

% !TEX root = main.tex

\section{Proof of Lemma~\ref{lem:1ARA}}  \label{sec:1NRA}

%In this and in the next section we  prove the first main result:
%% a generalization of Parikh theorem for register automata with one register:
%
%%\begin{theorem} \label{thm:1ARA}
%%Parikh images of \kNRA 1 languages are rational.
%%\end{theorem}
%
%We actually prove a slightly refined result (needed later for grammars).
%Denote by $\langoffull \A q a {q'} {a'}$ the set of data words input by some run starting
%in $\conf q a$ and ending in $\conf {q'} {a'}$.
%%
%\begin{lemma} \label{lem:1ARA}
%The languages $\langoffull \A q a {q'} {a'}$ have rational Parikh images.
%\end{lemma}
%%
%\noindent
%Lemma~\ref{lem:1ARA} implies Theorem~\ref{thm:1ARA}, as $\langof \A$ is an orbit-finite union:
%\[
%\langof \A \ = \ \bigcup_{q\in I, q'\in F, a,a'\in \atoms} \langoffull \A q a {q'} {a'}.
%\]

%Furthermore, w.l.o.g.~we assume that each location admits at most one type of outgoing transition rules.
%This can be easily achieved by appropriate duplication of locations.
%Therefore we can speak about register-preserving and register-updating location as well.

Consider a fixed \kNRA 1 $\A = \tuple{H, Q, I, F, \Delta}$.

\para{Proof strategy}
The proof proceeds by a sequence of simplifying steps, as stated in consecutive 
Lemmas~\ref{lem:LF}--\ref{lem:LF2} in this section and
in Lemmas~\ref{lem:anti0}--\ref{lem:D} in the next one.
Instead of only considering Parikh images of input words, in the proof we
investigate Parikh images of \emph{runs}, mostly concentrating 
on \emph{alterings}  of register value along a run.  
This leads us to consider, besides languages over the alphabet $H\times\atoms$ of a \kNRA 1, 
also languages over richer alphabets:
\begin{itemize}
\item languages of \emph{altering paths} over the alphabet $(Q\times\atoms\times Q) \ \cup \ (H\times \atoms)$ in Lemma~\ref{lem:LF};
\item languages of \emph{altering loops} over the alphabet $\atoms^2 \times \atoms$ in Lemma~\ref{lem:LF2};
\item languages of \emph{anti-paths} and \emph{anti-cycles} over $\atoms \times \pow 2 \atoms$ 
in Lemmas~\ref{lem:anti0}--\ref{lem:D}.
\end{itemize}
The intuitive idea underlying the final, most technical steps (Lemmas~\ref{lem:anti0}--\ref{lem:D})
is, roughly speaking,
that the set of words 
\[
\pair {a_1} {b_1}\,
\pair {a_2} {b_2}\,
\ldots\,
\pair {a_n} {b_n}
\]
over $\atoms^2$,
satisfying $b_i \neq a_{i+1}$ for all $i = 1, \ldots, n-1$, 
%(taking $n+1$ as $1$), 
has rational Parikh image.
Notably, this is not true for \emph{paths}, where one requires $b_i = a_{i+1}$ instead.

\subsection{Proof of Lemma~\ref{lem:1ARA}}

For locations $q, p\in Q$ of $\A$ and  $a\in\atoms$, let $L_{q a p}$ be the language of all data words
read by a run from configuration $\conf q a$ to $\conf p a$ that use register-preserving transitions only
(thus the register stores $a$ along the whole run).
\begin{lemma} \label{lem:Lqap}
The languages $L_{q a p}$ are rational.
\end{lemma}
\begin{proof}
We only need to consider register-preserving transitions.
Define the finite alphabet $\Delta = H\times\set{\fee,\fnne}$ and consider every transition rule 
$(\conf q x, \pair h y, \varphi, \conf {q'}{x'})$
%$(p', h, \varphi, q')$ 
to be \emph{labeled} by $\pair h \varphi \in \Delta$. 

Fix $q, a$ and $p$ and 
let $E_{qp}$ be the classical regular expression over $\Delta$ defining the labels of
all those runs from $q$ to $p$ that only use transitions of types \ee and \nne.
Then the language $L_{q a p}$ is defined by the expression $E_{q a p}$ obtained from $E_{q p}$ by
replacing $(h, \fee)$ with $\pair h a$ and replacing $\pair h \fnne$ with
\[
L_{e, \neq a} \quad = \quad \bigcup_{b\in\atoms - \set{a}} \pair h b.
\]
Thus $L_{q a p}$ is a rational data language.
\end{proof}
%
%Register automata \emph{alter} its register value, which is the main source of hardness of analysis thereof.
%
We now state the central lemma that generalises Example~\ref{ex:L1L2}.
Define the language $P$ over the alphabet $(Q\times\atoms\times Q) \ \cup \ (H\times \atoms)$  
%(in case of singleton $A$ it is essentially $L_1$ from Example~\ref{ex:L1L2}):
containing words of the form ($n \geq 1$):
\begin{align} 
\begin{aligned} \label{eq:ap}
\tuple{q_1, a_1, p_1} \tuple{h_1, b_1} \tuple{q_2, a_2, p_2} \tuple{h_2, b_2} \ldots  \tuple{q_n, a_n, p_{n}} 
\end{aligned}
\end{align}
such that 
%$n \geq 1, a_1 \neq a_2 \neq \ldots \neq a_n$. $(a_i, b_i, a_{i+1}) \models \varphi(q_i)$
$\conf {p_i} {a_i} \trans {\pair {h_i} {b_i}} \conf {q_{i+1}}{a_{i+1}}$ is a register-updating transition for $i = 1, \ldots, n-1$
(in particular $a_i \neq a_{i+1}$ for $i = 1, \ldots, n-1$).
Words in $P$ we call \emph{altering paths}.
Furthermore, define the subsets $P_{\conf q a \, \conf {q'} {a'}} \subseteq P$ of those altering paths
as in~\eqref{eq:ap} where $\conf q a = \conf {q_1} {a_1}$ and
$\conf {q'}{a'} = \conf {p_n}{a_n}$.
%that start in location $q_1 = q$, end in location $p_{n} = q'$, and where the first atom
%is $a_1 = a$ and the last one is $a_n = a'$.

\begin{lemma} \label{lem:LF}
Altering path languages $P_{\conf q a \, \conf {q'} {a'}}$ have rational Parikh images.
%, for every $a, a'\in\atoms$ and $q,q'\in Q$.
\end{lemma}
Before proving the lemma we use it to complete the proof of Lemma~\ref{lem:1ARA}.
Indeed, $\langoffull \A q a {q'} {a'}$ is obtained from the altering path language
$P_{\conf q a \, \conf {q'} {a'}}$
using the equivariant substitution ($q,p$ range over locations and $a,b$ over $\atoms$):
\[
\tuple{q, a, p} \quad \mapsto \quad L_{q a p} \qquad\qquad
\tuple{h, b} \quad \mapsto \quad \tuple{h, b}.
\] 
As a substitution by languages with rational Parikh images preserves rationality
of Parikh image (cf.~Lemma~\ref{lem:subst}), 
by Lemmas~\ref{lem:Lqap} and~\ref{lem:LF} we deduce
that the languages $\langoffull \A q a {q'} {a'}$ have rational Parikh images, as required.
%
%\end{proof}

\subsection{Proof of Lemma~\ref{lem:LF}}

%Towards proving Lemma~\ref{lem:LF} 
We define, for a 
register-updating transition constraint $\varphi \in \set{\fen, \fnnn}$
%rule $t = (p, e, \varphi, q)$ 
and (not necessariy distinct) atoms $a',b,a\in\atoms$, the language 
$L_{(a',b) \varphi a}$ over the alphabet $\atoms^2 \times \atoms$ as follows: 
let $L_{(a'_0,b_0) \varphi a_{n+1}}$ contain all (possibly empty) words of the form
\begin{align}  \label{eq:alter}
\tuple{\tuple{a_1, a'_1},\, b_1}\,\, \tuple{\tuple{a_2, a'_2}, \, b_2} \,\, \ldots  \,\, \tuple{\tuple{a_n, a'_n}, \, b_n} 
\end{align}
such that $(a'_i, b_i, a_{i+1}) \models \varphi$ for $i = 0, \ldots, n$.
We omit the case $\varphi = \fne$ as it is can be treated symmetrically to the case $\varphi = \fen$.
Words in $L_{(a'_0,b_0) \varphi a_{n+1}}$  we call \emph{altering loops}.
Intuitively, a letter $\tuple{\tuple{d, d'}, e}\in \atoms^2\times\atoms$ represents
(cf.~the substitution~\eqref{eq:subst} below), for some locations $p', p$ and $h\in H$,  an altering path from
$\conf p d$ to $\conf {p'} {d'}$ followed by a register-updating transition that inputs $\pair h e$.
We derive Lemma~\ref{lem:LF} from the following one (proved itself in Section~\ref{sec:LF2} below):
\begin{lemma} \label{lem:LF2}
Altering loop languages  $L_{(a',b) \varphi a}$ have rational Parikh images.
%, for every $a', b,a\in\atoms$ and register-updating constraint $\varphi$.
\end{lemma}

%Before proving the lemma we use it to prove Lemma~\ref{lem:LF}.
%
We mimic the standard proof of Kleene theorem, 
exploiting altering loops to capture all iterations along loops in $\A$.
We proceed by induction on the number of 
register-updating transition rules in $\A$.
If there is no such transition rules, we have trivial (and obviously rational) altering path languages
\[
P_{\conf q a \, \conf {q'} {a'}} = \begin{cases}
\set{\tuple{q, a, q'}} & \text{ if } a = a',\\
\emptyset & \text{ otherwise.}
\end{cases}
\]

%rational by Lemma~\ref{lem:Lqap}.
Otherwise, remove an arbitrary register-updating transition rule 
$t = (\conf q x, \pair h y, \varphi, \conf {q'}{x'})$
%$t = (p', h, \varphi, p)$ 
from $\A$
(if $\varphi = \fne$ consider the inverse of $\A$ and $\varphi = \fen$ instead),
and use the induction assumption for the so obtained automaton $\A'$ to get altering path languages
$K_{\conf q a \, \conf {q'} {a'}}$ for every locations $q, q'$ and atoms $a, a'$, with rational Parikh images.
Let $L_{(c',b) h \varphi c}$ be the language obtained from the altering loops 
$L_{(c',b) \varphi c}$ by the equivariant substitution 
($d, d', e$ range over $\atoms$)
\begin{align} \label{eq:subst}
\tuple{\tuple{d,d'},e} \quad \mapsto \ K_{\conf p d \, \conf{p'}{d'}} \  \tuple{h,e}.
\end{align}

Rationality of the Parikh images of the altering path languages $P_{\conf q a \, \conf {q'} {a'}}$ of $\A$ follows by the fact that
$P_{\conf q a \, \conf {q'} {a'}}$ is equal to the union of
$K_{\conf q a \, \conf {q'} {a'}}$ and the following set
\begin{align}   %\label{eq:Kleene0}
%& \bigcup_{\footnotesize \begin{array}{l} c',b,c\in\atoms\\ (c',b,c)\models\varphi\\\phantom \ \end{array}}
%K_{\conf q a \, \conf {p'} {c'}} \, \tuple{h,b}\, 
%%L_{(c'',b) e \varphi c} \, 
%K_{\conf p c \, \conf {q'} {a'}}  \\
%&
 \label{eq:Kleene1}
\!\! \bigcup_{c',b,c\in\atoms}
%\\ (c',b,c)\models\varphi\end{array}
K_{\conf q a \, \conf {p'} {c'}} \, \tuple{h,b}\, 
L_{(c',b) h \varphi c} \, K_{\conf p c \, \conf {q'} {a'}}.
\end{align}
To show the equality, we observe that $K_{\conf q a\, \conf{a'}{q'}}$ contains all % altering paths in $\A'$, i.e., 
altering paths in $\A$ that do not use $t$, and claim that
the set~\eqref{eq:Kleene1} contains those altering paths in $\A$ that do use $t$.
Specifically,
as $\varepsilon \in L_{(c',b) h \varphi c}$, we obtain altering paths using $t$ exactly once
(dotted arrow depict altering paths in $\A'$):
\vspace{-2mm}
\begin{figure}[H] %[btp]
\begin{center}
\includegraphics[width=5.2cm]{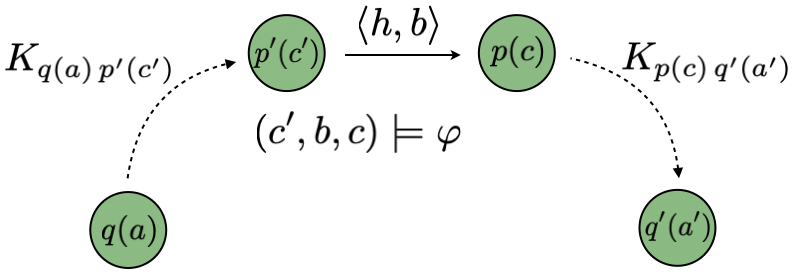}
\end{center}
\end{figure}
\vspace{-4mm}
%
%
%\noindent
%These altering paths factorise into a prefix before $t$ (altering paths
%from $\conf q a$ to $\conf {p'} {c'}$, for some $c'\in\atoms$), a single use of $t$, and
%the suffix after $t$ (altering paths from $\conf p c$ to $\conf {q'} {a'}$, for some $c\in\atoms$).
%%
%Finally, the set~\eqref{eq:Kleene1} contains those altering paths in $\A$ that use $t$ 

\noindent
or more than once (for instance twice, as shown in the figure):
\vspace{-6mm}
\begin{figure}[H] %[btp]
\begin{center}
\includegraphics[width=5.7cm]{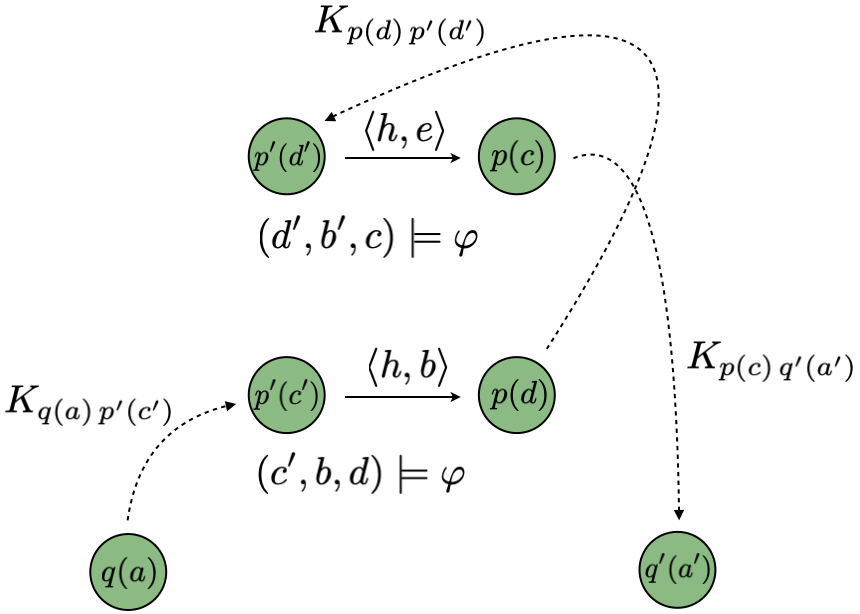}
\end{center}
\end{figure}
\vspace{-6mm}

\noindent
In general, a word in~\eqref{eq:Kleene1} factorises into a prefix before the first use of $t$ 
(an altering path from $\conf q a$ to $\conf {p'} {c'}$),
the suffix after the last use of $t$ (an altering path from $\conf p c$ to $\conf {q'} {a'}$),
and the infix leading from $\conf {p'} {c'}$ to $\conf {p}{c}$.
The infix starts with the letter $\tuple{h, b}$ input by the first traversal of $t$,
and then contains alternately altering paths
that do not use $t$ (from $\conf p d$ to $\conf {p'} {d'}$, %described by the sets $K_{\conf {p} {d}\, \conf {p'} {d'}}$ 
for some  $d,d'\in\atoms$)
and traversals of $t$ (a letter $\tuple{h,e}$ for some $e\in\atoms$), cf.~the substitution~\eqref{eq:subst}.
Therefore, by the definition~\eqref{eq:alter} of altering loops $L_{(a',b) \varphi a}$,
the set~\eqref{eq:Kleene1} contains exactly those altering paths in $\A$ that do use $t$, as claimed.

\subsection{Proof of Lemma~\ref{lem:LF2}} % (case $\varphi = \fnnn$)}
\label{sec:LF2}

We concentrate on the hardest case $\varphi = \fnnn$ (all the three atoms involved in $\varphi$
are pairwise distinct). 
The remaining case $\varphi = \fen$ is obtained then using the substitution
\[
\tuple{\tuple{a, a'}, b} \quad \mapsto \quad \tuple{\tuple{a, a'}, a'}.
\]
%(The remaining two cases are similar but easier; they are dealt with in~Section~\ref{sec:LF} in the appendix.)
%
We need to show that the altering loop languages $L_{(c,b) {\fnnn} a}$ have rational Parikh images.
Recall that
$L_{(c_0,b_0) {\fnnn} a_{n+1}}$ contains all words over $\atoms^2 \times \atoms$ of the form
\begin{align}  \label{eq:alt}
\tuple{\tuple{a_1, c_1},\, b_1}\,\, \tuple{\tuple{a_2, c_2}, \, b_2} \, \, \ldots  \,\, \tuple{\tuple{a_n, c_n}, \, b_n}
\end{align}
such that 
%(1) $a_1 \notin\set{c, b}$ and 
%(2) 
$c_i, b_i, a_{i+1}$ are pairwise different for $i = 0, \ldots, n$.
%(3) $a \notin \set{c_n, b_n}$. 
%The conditions (1) and (3) we call the beginning and end constraints, respectively.
 
We reduce Lemma~\ref{lem:LF2}  to Lemma~\ref{lem:anti0} (which constitutes the technical core
of the proof of Lemma~\ref{lem:1ARA}).
Relying on the observation that $b_i$ and $c_i$ play entirely symmetric roles in~\eqref{eq:alt}
and are forcedly distinct, 
we rearrange words~\eqref{eq:alt} into words over the alphabet 
$\Gamma = \atoms \times \pow 2 \atoms$ as follows:
\begin{align}  \label{eq:alt2}
\tuple{a_1, \set{b_1, c_1}}\,\, \tuple{a_2, \set{b_2, c_2}} \,\, \ldots  \,\, \tuple{a_n, \set{b_n, c_n}}.
\end{align}
Let $\ap {b_0} {c_0} {a_{n+1}} \subseteq \Gamma^*$ denote the language of all 
\emph{nonempty} words of the form~\eqref{eq:alt2}
subject to the same constraints as in~\eqref{eq:alt}, namely $a_{i+1}\notin\set{b_i,c_i}$ for $i = 0, \ldots, n$;
these words we call \emph{anti-paths} in the sequel.
Note that $a_i\in \set{b_i, c_i}$ is allowed. %for $i = 1, \ldots, n$.
We observe that $L_{(c,b) {\fnnn} a}$ is obtained from $\ap b c a$ using the equivariant substitution
\[
\tuple{d, \set{e,f}} \mapsto \tuple{d,e} \, f \ \ \cup \ \ \tuple{d, f} \, e,
\]
and adding the empty word.
Therefore 
the language $L_{(c,b) {\fnnn} a}$ has rational Parikh image assuming $\ap b c a$ has so,
and Lemma~\ref{lem:LF2} is implied by the following core technical result:
\begin{lemma}%[Anti-path Lemma] 
\label{lem:anti0}
The anti-path languages $\ap b c a \subseteq \Gamma^*$ have rational Parikh images.
\end{lemma} 
\noindent
The proof of Theorem~\ref{thm:1ARA} is thus completed once we prove Lemma~\ref{lem:anti0}.
The whole next section is devoted to this task.
%This is the subject of the next section.

\section{Anti-paths: Proof of Lemma~\ref{lem:anti0}}
\label{sec:anti}

For a letter $\alpha = \tuple{a, \set{b,c}} \in \Gamma$ we call the atom $a$ its \emph{source}, 
and the two-element set $\set{b,c}$ its \emph{target}, denoted $a= \src \alpha $ and $\set{b,c} = \trg \alpha$,
respectively. 
For a word $w = \alpha_1 \ldots \alpha_n \in\Gamma^*$ we denote by 
$\src w = \src {\alpha_1}$ the first source, and by
$\trg w = \trg {\alpha_n}$ the last target.

\para{Anti-cycles}
An anti-path $w$ is called an \emph{anti-cycle} if $\src w \notin \trg w$
(the first source does not belong to the last target).
Anti-cycles are closed under cyclic shifts, and hence we use the cyclic order when
speaking about precedence of letters in anti-cycles. 
Denote the set of all anti-cycles by $\acallgamma$.
We build on a simple but crucial observation: 
anti-paths $\ap b c a$ are exactly those words $w\in\Gamma^*$ which, prolonged with
a single letter $w \, \pair a {\set{b,c}} \in \Gamma$, form an anti-cycle:
\begin{claim} \label{claim:ac}
$\ap b c a = \setof{w\in\Gamma^*}{w \, \pair a {\set{b,c}} \in \acallgamma}$.
\end{claim}
\begin{lemma}%[Appendix~\ref{sec:anti-app}]   
\label{lem:minus}
If $L \subseteq \Gamma^*$ is rational and $\alpha\in\Gamma$ then the language
$L \minu \alpha = \setof{w\in\Gamma^*}{w \, \alpha \in L}$ is rational too.   
\end{lemma}
Claim~\ref{claim:ac} and Lemma~\ref{lem:minus} prove Lemma~\ref{lem:anti0}, once we have:
\begin{lemma} \label{lem:C}
$\acallgamma$ has rational Parikh image.
\end{lemma}
Indeed, let $L\subseteq \Gamma^*$ be rational and Parikh-equivalent to $\acallgamma$.
By Claim~\ref{claim:ac}, $\ap b c a$ is Parikh-equivalent to  $L \minu  \tuple{a, \set{b,c}}$, which is rational by 
Lemma~\ref{lem:minus}.
Thus it suffices to prove Lemma~\ref{lem:C}.

We mostly focus on a special but central case of Lemma~\ref{lem:C}, namely
we restrict to the sub-alphabet
\[
\Sigma \ = \ \setof{\alpha \in \Gamma}{\src{\alpha}\notin \set{b,c}} \ \subseteq \ \Gamma.
\]
%(In Appendix~\ref{sec:ap} we consider back the original language $\acallgamma$, 
%thus completing the proof of Lemma~\ref{lem:C}.)
%
\begin{lemma}%[Appendix~\ref{sec:anti-app}]  
\label{lem:DC}
If the language $\acall = \acallgamma\cap \Sigma^*$ has rational Parikh image
then $\acallgamma$ has rational Parikh image too.
\end{lemma} 
\begin{lemma} \label{lem:D}
The language $\acall$ has rational Parikh image.
\end{lemma} 
%
%\noindent

\subsection{Proof of Lemma~\ref{lem:D}}

%\begin{proof}[]
%Let $a\in\atoms$ and $\set{b,c}\in\pow 2 \atoms$ be fixed throughout the proof.
%To simplify the presentation we assume $a\notin\set{b,c}$.

The number of different sources of letters appearing in a data vector $v : \Sigma \to \N$,
i.e., the size of the set
\begin{align} \label{eq:sources}
\srcs v \ = \ \setof{\src \alpha}{\alpha \in \dom v},
\end{align}
we denote by $\ord v$ and call \emph{the order} of $v$
(clearly, an atom can be the source of more than one letter in $\dom v$). 
The order of a data word $w\in\Sigma^*$ is defined naturally as $\ord w = \ord{\Par w}$.
We write $\acallpar {<n}$ (resp.~$\acallpar {\geq n}$) for the subsets of $\acall$ containing anti-cycles  
%We write $\ordleq \Sigma n$ (resp.~$\ordgre \Sigma n$) for the set of data vectors 
of order smaller than $n$ (resp.~at least $n$).
Anti-cycles of bounded order can be easily dealt with separately:
\begin{lemma}%[Appendix~\ref{sec:anti-app}]   
\label{lem:belown}
For every $n\in\N$, the language $\acallpar {<n}$
% the set of all anti-paths of order at most $n$ 
is rational.
\end{lemma}
Therefore, in the rest of the proof we concentrate on anti-cycles or order at laest $n$, for a
sufficiently large $n\in\N$.

\para{Source graphs}
In the sequel we consider directed graphs without self-loops or parallel edges, but possibly containing 
tight two-vertex cycles.

%Given $X\subseteq \dom v$, the \emph{restriction} of $v$ to $X$ is the data vector that differs from $v$ by
%mapping letters $\alpha\notin X$ to $0$. 
%Guided by the crucial property of anti-paths that the source of every letter does not belong to the target of the preceding letter,
%we define $\precn d : \Sigma \to \N$ to denote
%the restriction of $v$ to 
%$$\setof{\alpha \in \dom v}{d\notin \trg \alpha, d \neq \src \alpha}.$$
%Thus $\precn d$ contains all letters of $v$ with source different than $d$ that can \emph{precede} a 
%$d$-sourced letter in an anti-path.
%
%For a data vector $v : \Sigma \to \N$ w

Let $v : \Sigma \to \N$ be a fixed data vector.
Guided by the crucial property of anti-paths that the source of every letter does not belong to the target of the preceding letter,
we define the directed graph $\gr v = (\srcs v, E_v)$, called \emph{source graph} induced by $v$:
let the vertices $\srcs v$ of $\gr v$ be the sources of all letters appearing in $v$, as defined in~\eqref{eq:sources},
and let $(d,e)\in E_v$ be an edge if, and only if
% $d = \src \alpha$ for some $\alpha \in \precn e$. Equivalenty, $(d,e)\in E_v$, for distinct atoms $d\neq e$, if
\[
%E_v \ = \ \setof{(d,e)}{
\prettyexists{\alpha \in \dom v}{d = \src \alpha, e\notin \trg \alpha}.
\]
Whenever $(d,e)\notin E_v$, for distinct atoms $d\neq e$, we say that
$d$ \emph{excludes} $e$ (or call $(d,e)$ an \emph{excluded} edge); 
equivalently, $e$ belongs to the target of every letter in 
$\dom v$ with source $d$:
\[
\prettyforall {\alpha \in \dom v}{\, d = \src \alpha \implies e\in \trg \alpha}.
\]
Note that an atom never excludes itself, due to restriction to $\Sigma$, and
that $\gr v$ depends only on the set $\dom v \subseteq \Sigma$ of letters appearing in $v$, 
and not on cardinalities of letters in $v$.

Let  $\inn e = \setof{d \in \srcs v}{(d,e)\in E_v}$ denote the set of in-neighbours of a vertex $e$,
and let $\ind e = \size{\inn e}$ denote the \emph{in-degree} of $e$.
Symmetrically we define 
out-neighbours $\outn e$ and \emph{out-degree} $\outd e$.
Clearly, an atom may exclude at most two other atoms, and hence
(let $n=\ord v$):
\begin{claim} \label{claim:gr}
$\outd d \geq n-3$ for every vertex $d\in\srcs v$.  % , where $n=\ord v$.
%Therefore the total number of edges is at least $n(n-3)$.
\end{claim}
\begin{corollary} \label{cor:2n}
There are at most $2n$ excluded edges.
\end{corollary}
%
%Note that $\ind d = 0$ if, and only if $\precn d = \emptyset$.
%
In the sequel we rely on Claim~\ref{claim:gr} and Corollary~\ref{cor:2n} according to
% the fact that only few edges (at most $2n$) are excluded from $\gr v$ and hence 
which $\gr v$ is not much different from the full directed clique.

For $A\subseteq \dom v$, let $\restr A v$ denote the restriction of $v$ to $A$:
$\restr A v(\alpha) \ = \  v(\alpha)$ if  $\alpha \in A$, and $\restr A v(\alpha) = 0$ otherwise.
In the proof of Lemma~\ref{lem:above} we transform cycles in $\gr v$ into anti-cycles, using the following lemma:
\begin{lemma} \label{lem:2ap}
For every simple cycle $\pi = a_1 a_2 \ldots a_n$ in $\gr v$ there exists an anti-cycle $w$
with $\Par w = \restr A  v$ where $A = \setof{\alpha\in\dom v}{\src \alpha \in \set{a_1, a_2, \ldots, a_n}}$.
%and $\src w = a_1$.
\end{lemma}
\begin{proof}
We arrange the letters into an anti-path $w$ by taking first all $a_1$-sourced letters in a consecutive block, 
then all $a_2$-sourced ones in a consecutive block, etc.
The order of $a_i$-sourced letters inside a block (including repetitions of equal letters) 
is irrelevant as long as the last one, 
say $\alpha$, satisfies $a_{i+1} \notin \trg \alpha$ (where $n+1$ is identified cyclicly with $1$).
\end{proof}

In the proof of Lemma~\ref{lem:above} we also  use a sufficient condition for a directed graph to admit a Hamiltonian cycle:
\begin{theorem}[\cite{Ghouila-Houri}, cf.~ also Thm.~1 in~\cite{Kuhn-Osthus}] \label{thm:Hamilton}
Let $\G$ be a strongly connected directed graph with $n$ vertices such that for every two vertices
$d, d'$, $\ind{d} + \outd{d'} \geq n$.
Then $\G$ contains a Hamiltonian cycle.
\end{theorem}
The tool will be applicable due to the following observation:
\begin{lemma} \label{lem:sc}
For sufficiently large $n$, 
a directed graph with $n$ vertices such that $\ind d\geq 3$ and $\outd d \geq n-3$ for every vertex $d$,
is necessarily strongly connected.
\end{lemma}
\begin{proof} % [Proof of Lemma~\ref{lem:sc}]
Consider the decomposition of the graph into strongly connected components.
As the first step we observe that there may be no singleton components $\set{d}$.
Indeed, by the assumption we have $\ind d + \outd d \geq n$,
and hence $d$ forms a tight 2-vertex cycle with some other vertex $d'$.

In the sequel we use Corollary~\ref{cor:2n}.
As the second step we argue that for sufficiently large $n$,
a component $\set{d,e}$ of size $2$ is impossible
(and, in consequence, a component of size $n-2$ is impossible too).
Towards contradiction, suppose $\set{d,e}$ is a strongly connected component (hence
the two vertices form a tight cycle).
In consequence,
(a) the sets $V_d = \outn d - \set{e}$ and $V_e = \inn e - \set{d}$ are disjoint, and
(b) there is no edge from $V_d$ to $V_e \cup \set{d,e}$.  %, and (c) there is no edge from $V_d$ to $\set{d,e}$
As $\outd d\geq 3$ and $\ind e\geq 3$,
we have $\size{V_d}\geq n-4$ and $\size{V_e}\geq 2$.
By (a) we deduce $\size{V_d} = n-4$ and $\size {V_e} = 2$.
%which implies $\size{V - V_d} = 4$.
By (b), all $4(n-4)$ edges from $V_d$ to $V_e\cup\set{d,e}$ are excluded.
This is impossible as long as $4(n-4)>2n$.

Likewise one argues that 
there may be no component of size strictly between $2$ and $n-2$.
Indeed, supposing there is a component $C$ of size $k$, for $2 < k < n-2$, no vertex in $C$ may form a tight
cycle with other vertex outside of $C$, and hence at least $k(n-k)$ edges are excluded.
This is impossible as long as $k(n-k) > 2n$.
As $k(n-k)$ reaches its minimum for $k=3$ or $k=n-3$, there may be no component of size strictly
between $2$ and $n-2$ as long as
$
3(n-3) > 2n.
$
\end{proof}
%As $\outd d + \ind e \geq n$, we conclude that

\para{Non-degeneracy}
Let $\precn d = \setof{\alpha\in\dom v}{d\neq \src \alpha, d\notin \trg \alpha}$
denote the set of letters that can precede a $d$-sourced letter and have themselves source different than $d$.
A data vector $v : \Sigma \to \N$ is called \emph{non-degenerate} 
%(and \emph{degenerate} otherwise)
if the following conditions holds:
\begin{itemize}
\item[(1)] $\inn d \neq \emptyset$ for every $d\in \srcs v$, 
\item[(2)] $\inn d \cup \inn e \not\subseteq \set{d,e}$ for every non-equal $d,e\in \srcs v$, 
\item[(3)] $\size{\restr {\precn d \cup \precn e} v} \geq 2$ for every non-equal $d,e\in \srcs v$.
%\item[(3)] $\size{\precn d + \precn e} \geq 2$ \todooo{to corr} for every non-equal $d,e\in \srcs v$. 
\end{itemize}
(1) excludes vertices of in-degree $0$.
(2) excludes pairs of vertices $d,e$ with $\inn d = \set{e}$ and $\inn e = \set{d}$.
(3) excludes the case when there is only
one letter $\alpha\in\dom v$ that can precede $d$- or $e$-sourced letters,
and moreover $v(\alpha) =1$.
\begin{lemma} \label{lem:above}
For data vectors $v : \Sigma\to\N$ of sufficiently large order,
$v \in \Par {\acall}$ if, and only if $v$ is non-degenerate.
\end{lemma}
\begin{proof}
Let $\gr v = (\srcs v, E_v)$ be the source graph and $n = \ord v$. % By assumption, $n > 18$.

The 'only if' implication is immediate for data vectors of order at least $3$.
Indeed, suppose $v = \Par w$ for an anti-cycle $w\in\acall$.
By the definition of anti-cycles, $\inn d \neq \emptyset$ for every $d\in\srcs v$ and hence (1) forcedly holds.
The other two conditions are easily shown by contradiction.
Indeed, if (2) fails for some $d,e\in\srcs v$ then every $d$- or $e$-sourced letter would be preceded in $w$ by a
$d$- or $e$-sourced one, which is impossible as long as $\ord v \geq 3$.
Finally, if (3) fails then the same letter $\alpha$ would have to precede
two different letters in $w$. 

%\medskip

For the 'if' implication, we assume that $v$ is non-de\-ge\-ne\-rate ((1)--(3) hold)
and prove that $v = \Par w$ for some $w\in \acall$.

Let $k = 9$.  % be an integer whose value will be revealed later. 
Due to Corollary~\ref{cor:2n} we can assume $n$ to be large enough so that:
\begin{claim} \label{claim:atmost2}
At most two atoms in $\srcs v$  have in-degree $< k$.
\end{claim}
\noindent
In other words, this means that there are no $3$ atoms excluded by at least $n-k$ vertices.
Therefore, relying on Corollary~\ref{cor:2n} it is enough to assume
$
3(n-k) \ > \ 2n,
$
i.e., $n> 3k$.

Let $a_1, a_2 \in V$ be the vertices with the smallest in-degrees.
%$$\ind {a_1} \leq  \ind {a_2} \leq \ldots \leq  \ind{a_n}.$$
By assumption, $\ind {a_1} \geq 1, \ind {a_2} \geq 1$, and by Claim~\ref{claim:atmost2} we have:
%Taking into account removal of at most $2$ vertices  and
%by Claim~\ref{claim:atmost2}
%%
%As the total number of excluded edges is at most $2n$, at most $\frac 2 3 n$ vertices exclude $a_3$, and hence
%%
%\begin{claim}
%$\ind {a_3} \geq \frac 1 3 n -1$.
%\end{claim}
%%
%\noindent
%In the sequel we assume that $n$ is large enough
%(i.e., $n\geq 24$) so that $\ind {a_3}  \geq 7$.
%We thus have:
%
\begin{claim} \label{claim:ind}
Every $d\in \srcs v - \set{a_1, a_2}$ satisfies $\ind d \geq k$.
\end{claim}

%We use a pointwise order $\sqsubseteq$ on data vectors.
%We construct an anti-path $\bar w$ such that $\bar v = \Par {\bar w} \sqsubseteq v$ and such that
%$\ind {\src {\bar w}} \geq k$ and 
%all vertices $d$ not appearing as a source in $\bar w$ satisfy $\ind d \geq k$.

We construct a cycle $\pi$ in $\gr v$ such that ($\circ$) its first vertex $d$, as well as
vertices $d$ not contained in $\pi$, satisfy $\ind d \geq k$.
Due to (1), it suffices to consider the following cases:

\parainproof{Case 1. $\size{\inn {a_1} \cup \inn {a_2}} \geq 2$}

Relying on (1), choose in $\srcs v-\set{a_1,a_2}$ two distinct atoms $d \neq d'$ with $d \in \inn{a_1}$ and $d' \in \inn{a_2}$.
Due to (2) the atoms can be chosen so that $d \neq a_2$ or $d' \neq a_1$.
By symmetry we assume w.l.o.g.~that $d \neq a_2$.
If $d' = a_1$ we take the following simple path $\pi$ in $\gr v$ satisfying ($\circ$):
\begin{align*} %\label{eq:path11}
\xymatrix{
d \ar[r] & a_1 \ar[r] & a_2
}
\end{align*}
Otherwise, 
suppose $d' \neq a_1$ either.
By Claim~\ref{claim:ind}, $\ind {d}\geq k$ and $\ind {d'} \geq k$.
Choose in $V-\set{d,a_1,d',a_2}$ any atom $e$  with $e \in \inn{d'} \cap \outn {a_1}$
(since $\ind {d'}\geq k$, such $e$ exists as $a_1$ excludes at most two atoms, as long as $k\geq 7$). 
This yields the following simple path $\pi$ in $\gr v$ satisfying ($\circ$):
\begin{align*} % \label{eq:path12}
\xymatrix{
d \ar[r] & a_1 \ar[r] & e \ar[r] & d' \ar[r] & a_2
}
\end{align*}
 
\parainproof{Case 2. $\inn {a_1} = \inn {a_2} = \set{d}$ for some $d\in \srcs v{-}\set{a_1,a_2}$}

Take some two letters $\alpha_1, \alpha_2$ appearing in $v$ such that
$a_1\notin\trg{\alpha_1}$ and $a_2\notin\trg{\alpha_2}$.
Due to (3) we can assume that either $\alpha_1\neq\alpha_2$, or
$\alpha_1 = \alpha_2$ but $v(\alpha_1) \geq 2$ (their cardinality in $v$ is at least $2$).
Note that $\src {\alpha_1} = \src{\alpha_2} = d$,
and by Claim~\ref{claim:ind}, $\ind d \geq k$.
Choose in $V-\set{d,a_1,a_2}$ any atom $e$  with $e \in \inn{d} \cap \outn {a_1}$
(similarly as before, such $e$ exists as long as $k\geq 6$). 
This yields the non-simple path $\pi$
 in $\gr v$ satisfying ($\circ$):
\begin{align*} % \label{eq:path2}
\xymatrix{
d \ar[r] & a_1 \ar[r] & e \ar[r] & d \ar[r] & a_2
} 
\end{align*}
%
%\medskip

\medskip
We have thus constructed a path $\pi$ from $d$ to $a_2$.
If $a_2 \notin \inn d$, append at the end of $\pi$ any vertex $c$ such that $c\in \outn {a_2} \cap \inn d$.
As before, such a vertex exists since $a_2$ excludes at most $2$ atoms and $\ind d \geq k$,
as long as $k \geq 8$.
Therefore the last vertex $c$ of $\pi$ satisfies $c\in\inn d$, which means that $\pi$ is a cycle
as required.

\medskip

In Case 1 we transform $\pi$, using Lemma~\ref{lem:2ap}, into an anti-cycle $\bar w$.
In Case 2 we proceed similarly, except that the vertex $d$ appears twice in $\pi$; this exception is treated by splitting
all $d$-sourced letters into two disjoint blocks (cf.~the proof of Lemma~\ref{lem:2ap}),
containing $\alpha_1$ and $\alpha_2$, respectively.

We now remove, intuitively speaking, the anti-cycle $\bar w$ from $v$ thus obtaining a smaller data vector
$v'$ to which we apply  Theorem~\ref{thm:Hamilton} and Lemma~\ref{lem:2ap}.
%Let $d = \src{\bar w}$ and let $\tuple{a_2, \set{e,e'}}$ be the last (necessarily $a_2$-sourced) letter  in $\bar w$.
We remove from $v$ all letters appearing in $\bar w$, and add a single letter 
$\beta = \tuple{\src{\bar w}, \trg{\bar w}}\in\Sigma$.
This yields a data vector $v'$.
As the length of $\pi$ is at most $6$,
the in-degree of a node $e$ in the graph $\gr {v'}$ may  be smaller by at most $6$ than
in the graph $\gr v$.
Thus $\ind e \geq 3$ in $\gr {v'}$ as $k\geq 9$.
Moreover $\outd e \geq n'-3$ in $\gr {v'}$, where $n'$ is the number of nodes of $\gr {v'}$, by Corollary~\ref{cor:2n}.
Therefore the graph $\gr {v'}$, assuming $n$ to be sufficiently large,
 satisfies assumptions of Lemma~\ref{lem:sc}, by which
$\gr {v'}$ is strongly connected.
In consequence, $\gr {v'}$
satisfies assumptions of Theorem~\ref{thm:Hamilton}, by which we derive a Hamiltonian cycle $\C$ in $\G$.
The Hamiltonian cycle is turned, using Lemma~\ref{lem:2ap}, into an anti-cycle in $w'$
with $\Par {w'} = v'$. Finally, replacing the letter $\beta$ in $w'$ by  $\bar w$, yields
an anti-cycle $w$ with $\Par w = v$, as required.
\end{proof}

Let $\ndegall$ denote the set of all non-degenerate data vectors,
and $\ndegn n = \setof{v\in\ndegall}{\ord v\geq n}$.
In these terms, Lemma~\ref{lem:above} claims $\ndegn n = \Par{\acallpar{\geq n}}$
for sufficiently large $n$.
%
%Relying on Claim~\ref{claim:ac}, according to which
%\[
%\Par{\kl b c a}\ = \ \setof{v:\Sigma\to\N}{v + \pair a {\set{b,c}} \in \Par \acall},
%\]
%and on Lemma~\ref{lem:above}, we define:
%\[
%\ndegall \  = \ \setof{v: \Sigma\to\N}{v + \tuple{a, {\set{b,c}}} \text{ is non-degenerate}}
%\]
%and prove:
%
\begin{lemma}%[Appendix~\ref{sec:anti-app}]  
\label{lem:nondegrat}
$\ndegn n$ is rational, for sufficiently large $n\in\N$.
\end{lemma}
For $n\in\N$ sufficiently large for Lemmas~\ref{lem:above} and~\ref{lem:nondegrat} to hold,
%relying on Lemma~\ref{lem:above} 
we decompose the Parikh image of anti-cycles into
\[
\Par {\acall} \ = \ \Par{\acallpar{<n}}  \ \cup \ \ndegn n,
%\Par{\klordgre n b c a}
% \ndegordgre n b c a 
\]
both of them rational by Lemmas~\ref{lem:belown} and Lemma~\ref{lem:nondegrat}, respectively.
Lemma~\ref{lem:D} is thus proved.

%
%\end{proof}

% !TEX root = main.tex

\section{Final remarks}
\label{sec:conc}

We have shown that Parikh images of languages of one-register automata are not semi-linear in general, but are rational;
and likewise for one-register context-free languages.
As a corollary of Theorem~\ref{thm:1G} we obtain an analog of Parikh's theorem mentioned in the introduction:
one-register context-free grammars are 
Parikh-equivalent to register automata (but not to one-register ones).
Indeed, %by induction on derivation one shows that 
every rational set of data vectors is the Parikh image of some register automaton.

%derivation trees of width bounded by $n$ can be simulated by automata with ${\cal O}(n)$ registers.

We conjecture that the restriction to one register can be dropped, and that general register context-free grammars
have rational Parikh images and are Parikh-equivalent to register automata;
our present proof techniques do not allow however to tackle the general case.
On the other hand our proof method routinely (but tediously) adapts to  \kCFG 1 of any arity,
but at the price of considering
anti-paths over a larger alphabet $\atoms \times \pow n \atoms$, where $n$ is the largest arity of a \kCFG 1.
%We postpone this tedious generalisation step to the extended version.

Besides dropping one-register restriction,
we envisage several potential directions of generalisation: 
richer input alphabets, % than $H\times\atoms$, 
more structured atoms, etc. 
As future work we leave also investigation of algorithmic problems on rational sets, 
like testing equality of such sets.
Finally, we hope to develop a
general theory of rational sets of data vectors, e.g., study closure properties, strictness of the star-height hierarchy, or logical characterisations.

\section*{Acknowledgment}
The authors would like to thank the anonymous reviewers for helpful remarks and suggestions.

\newpage



\bibliographystyle{IEEEtran}
\bibliography{citat}

\begin{thebibliography}{10}
\providecommand{\url}[1]{#1}
\csname url@samestyle\endcsname
\providecommand{\newblock}{\relax}
\providecommand{\bibinfo}[2]{#2}
\providecommand{\BIBentrySTDinterwordspacing}{\spaceskip=0pt\relax}
\providecommand{\BIBentryALTinterwordstretchfactor}{4}
\providecommand{\BIBentryALTinterwordspacing}{\spaceskip=\fontdimen2\font plus
\BIBentryALTinterwordstretchfactor\fontdimen3\font minus
  \fontdimen4\font\relax}
\providecommand{\BIBforeignlanguage}[2]{{%
\expandafter\ifx\csname l@#1\endcsname\relax
\typeout{** WARNING: IEEEtran.bst: No hyphenation pattern has been}%
\typeout{** loaded for the language `#1'. Using the pattern for}%
\typeout{** the default language instead.}%
\else
\language=\csname l@#1\endcsname
\fi
#2}}
\providecommand{\BIBdecl}{\relax}
\BIBdecl

\bibitem{KF94}
N.~Francez and M.~Kaminski, ``Finite-memory automata,'' \emph{Theor. Comput.
  Sci.}, vol. 134, no.~2, pp. 329--363, 1994.

\bibitem{S06}
L.~Segoufin, ``Automata and logics for words and trees over an infinite
  alphabet,'' in \emph{Proc.~{CSL} 2006}, ser. Lecture Notes in Computer
  Science, vol. 4207.\hskip 1em plus 0.5em minus 0.4em\relax Springer, 2006,
  pp. 41--57.

\bibitem{lmcs14}
M.~Boja\'nczyk, B.~Klin, and S.~Lasota, ``Automata theory in nominal sets,''
  \emph{Log. Methods Comput. Sci.}, vol.~10, no.~3, 2014.

\bibitem{atombook}
\BIBentryALTinterwordspacing
M.~Boja{\'n}czyk, ``Slightly infinite sets,'' a draft of a book. [Online].
  Available: \url{https://www.mimuw.edu.pl/~bojan/paper/atom-book}
\BIBentrySTDinterwordspacing

\bibitem{regexp-Kaminski}
M.~Kaminski and T.~Tan, ``Regular expressions for languages over infinite
  alphabets,'' \emph{Fundam. Informaticae}, vol.~69, no.~3, pp. 301--318, 2006.

\bibitem{regexp-Domagoj}
L.~Libkin, T.~Tan, and D.~Vrgoc, ``Regular expressions for data words,''
  \emph{J. Comput. Syst. Sci.}, vol.~81, no.~7, pp. 1278--1297, 2015.

\bibitem{regexp-Kurz}
A.~Kurz, T.~Suzuki, and E.~Tuosto, ``On nominal regular languages with
  binders,'' in \emph{Proc. {FOSSACS} 2012}, ser. Lecture Notes in Computer
  Science, L.~Birkedal, Ed., vol. 7213.\hskip 1em plus 0.5em minus 0.4em\relax
  Springer, 2012, pp. 255--269.

\bibitem{lics11}
M.~Boja\'nczyk, B.~Klin, and S.~Lasota, ``Automata with group actions,'' in
  \emph{Proc.~{LICS} 2011}, 2011, pp. 355--364.

\bibitem{Ghouila-Houri}
A.~Ghouila-Houri, ``Une condition suffisante d'existence d'un circuit
  hamiltonien,'' \emph{{C. R. Acad. Sci. Paris}}, vol.~25, pp. 495--497, 1960.

\bibitem{Kuhn-Osthus}
D.~K{\"u}hn and D.~Osthus, ``A survey on hamilton cycles in directed graphs,''
  \emph{European Journal of Combinatorics}, vol.~33, no.~5, pp. 750 -- 766,
  2012.

\bibitem{Parikh66}
R.~Parikh, ``On context-free languages,'' \emph{J. {ACM}}, vol.~13, no.~4, pp.
  570--581, 1966.

\bibitem{SI00}
H.~Sakamoto and D.~Ikeda, ``Intractability of decision problems for
  finite-memory automata,'' \emph{Theor. Comput. Sci.}, vol. 231, no.~2, pp.
  297--308, 2000.

\bibitem{NSV04}
F.~Neven, T.~Schwentick, and V.~Vianu, ``Finite state machines for strings over
  infinite alphabets,'' \emph{{ACM} Trans. Comput. Log.}, vol.~5, no.~3, pp.
  403--435, 2004.

\bibitem{MN-Kaminski}
N.~Francez and M.~Kaminski, ``An algebraic characterization of deterministic
  regular languages over infinite alphabets,'' \emph{Theor. Comput. Sci.}, vol.
  306, no. 1-3, pp. 155--175, 2003.

\bibitem{datamonoids}
M.~Boja{\'n}czyk, ``Data monoids,'' in \emph{Proc.~{STACS} 2011}, ser. LIPIcs,
  vol.~9.\hskip 1em plus 0.5em minus 0.4em\relax Schloss Dagstuhl -
  Leibniz-Zentrum f{\"{u}}r Informatik, 2011, pp. 105--116.

\bibitem{rigidMSO}
\BIBentryALTinterwordspacing
T.~Colcombet, C.~Ley, and G.~Puppis, ``Logics with rigidly guarded data
  tests,'' \emph{Log. Methods Comput. Sci.}, vol.~11, no.~3, 2015. [Online].
  Available: \url{https://doi.org/10.2168/LMCS-11(3:10)2015}
\BIBentrySTDinterwordspacing

\bibitem{BS20}
M.~Boja\'nczyk and R.~Stefa\'nski, ``Single-use automata and transducers for
  infinite alphabets,'' in \emph{Proc.~{ICALP} 2020}, ser. LIPIcs, vol.
  168.\hskip 1em plus 0.5em minus 0.4em\relax Schloss Dagstuhl -
  Leibniz-Zentrum f{\"{u}}r Informatik, 2020, pp. 113:1--113:14.

\bibitem{DL09}
S.~Demri and R.~Lazic, ``{LTL} with the freeze quantifier and register
  automata,'' \emph{{ACM} Trans. Comput. Log.}, vol.~10, no.~3, pp.
  16:1--16:30, 2009.

\bibitem{CM14}
T.~Colcombet and A.~Manuel, ``Generalized data automata and fixpoint logic,''
  in \emph{Proc.~{FSTTCS} 2014}, ser. LIPIcs, vol.~29.\hskip 1em plus 0.5em
  minus 0.4em\relax Schloss Dagstuhl - Leibniz-Zentrum f{\"{u}}r Informatik,
  2014, pp. 267--278.

\bibitem{CL15}
L.~Clemente and S.~Lasota, ``Reachability analysis of first-order definable
  pushdown systems,'' in \emph{Proc.~{CSL} 2015}, ser. LIPIcs, S.~Kreutzer,
  Ed., vol.~41.\hskip 1em plus 0.5em minus 0.4em\relax Schloss Dagstuhl -
  Leibniz-Zentrum f{\"{u}}r Informatik, 2015, pp. 244--259.

\bibitem{symbautom}
L.~D'Antoni and M.~Veanes, ``Minimization of symbolic automata,'' in
  \emph{Proc.~{POPL} '14}.\hskip 1em plus 0.5em minus 0.4em\relax {ACM}, 2014,
  pp. 541--554.

\bibitem{MSV03}
T.~Milo, D.~Suciu, and V.~Vianu, ``Typechecking for {XML} transformers,''
  \emph{J. Comput. Syst. Sci.}, vol.~66, no.~1, pp. 66--97, 2003.

\bibitem{dataautomata}
M.~Boja\'nczyk, C.~David, A.~Muscholl, T.~Schwentick, and L.~Segoufin,
  ``Two-variable logic on data words,'' \emph{{ACM} Trans. Comput. Log.},
  vol.~12, no.~4, pp. 27:1--27:26, 2011.

\bibitem{classautomata}
M.~Boja\'nczyk and S.~Lasota, ``An extension of data automata that captures
  {XPath},'' \emph{Log. Methods Comput. Sci.}, vol.~8, no.~1, 2012.

\bibitem{Pitts:book}
A.~M. Pitts, \emph{Nominal Sets: Names and Symmetry in Computer Science}, ser.
  Cambridge Tracts in Theoretical Computer Science.\hskip 1em plus 0.5em minus
  0.4em\relax Cambridge University Press, 2013, vol.~57.

\bibitem{Eilenberg74}
\BIBentryALTinterwordspacing
S.~Eilenberg, \emph{Automata, languages, and machines. {A}}, ser. Pure and
  applied mathematics.\hskip 1em plus 0.5em minus 0.4em\relax Academic Press,
  1974. [Online]. Available: \url{https://www.worldcat.org/oclc/310535248}
\BIBentrySTDinterwordspacing

\bibitem{Marta}
M.~Juzepczuk, ``Zbiory semiliniowe nad niesko{\'n}czonym alfabetem (in
  {Polish}),'' Master's thesis, University of Warsaw, 2013.

\end{thebibliography}

\newpage 

\appendix

% !TEX root = main.tex 

\subsection{Missing items in Section~\ref{sec:pre}}  \label{sec:pre-app}

\begin{proof}[\bf Proof of Lemma~\ref{lem:b-o-f}]
Fix an orbit-finite set $\Sigma$.
The 'only if' implication is immediate, as the length (or size) is invariant inside an orbit.
Towards the 'if' implication for data languages, we observe that the set $\Sigma^n$ of words of length $n$ is orbit-finite, 
for every $n\in\N$,
as Cartesian products preserve orbit-finiteness.
Therefore a language $X\subseteq \Sigma^*$ satisfying $\size v \leq n$ for $v\in X$, 
is a subset of a finite union of orbit-finite sets and hence orbit-finite itself.
In consequence, $\Par X$ is is also orbit-finite, as the image of $X$ under an equivariant function,
which proves the claim for sets of data vectors.
\end{proof}

\medskip

\subsection{Missing items in Section~\ref{sec:rat}}  \label{sec:rat-app}

\begin{lemma} \label{lem:L1}
The language $L_1$ from Example~\ref{ex:L1L2} is not rational.
\end{lemma}
\begin{proof}
Indeed, towards contradiction suppose $L_1$ is rational, and hence generated by a rational expression $R$.
Consider the sublanguage $L \subset L_1$ containing words in which all atoms are different. 
The language $L$ is orbit-infinite and hence it cannot be generated without star;
indeed, concatenation and orbit-finite sums preserve orbit-finiteness of languages.
Therefore, there must be a star subexpression $R'$ of $R$ such that the number of iterations of $R'$ is unbounded
in generation of words in $L$.
In other words, for every $n\in\N$ there is a word $w\in L$ whose some infix $u$ is generated by at least $n$ iterations of $R'$.
Thus $w = w' u w''$, the infix $u$ splits into $u = u_1 \ldots u_n$, and each of factors $u_i$ is generated by $R'$.
Choose $n$ sufficiently large, namely $n > 2\cdot \size{\supp{R'}}$
(considering union operations as atom-binding constructs, the support of a rational expression $R'$ consists of 
those atoms appearing in $R'$ which are not bounded by any union).
As no atom repeats twice in words in $L$, some of words $u_i$ is fresh for $R'$, i.e., $\supp u \cap \supp {R'} = \emptyset$,
and $u_i$ is either preceded or succeeded in $w$ by an atom $a\notin \supp{R'}$. 
Consider w.l.o.g.~the first case, and let $a'$ be the first atom in $u_i$. Necessarily $a\neq a'$.
As $R'$ is invariant under the swap  $a' \leftrightarrow a$, applying this swap to $u_i$ yields a word $u'$ still generated
by $R'$. Replacing $u_i$ by $u'$ in $w$ yields a word $w'$ still generated by $R$, but $w'\notin L_1$ as it contains
two consecutive atoms $a$.
The contradiction completes the proof.
%\todooo{is this proof correct?}
\end{proof}

\medskip

\subsection{Missing items in Section~\ref{sec:sl}}  \label{sec:sl-app}

\begin{proof}[\bf Proof of Proposition~\ref{prop:sh1}]
Every semilinear set is, by definition, a rational set of star-heigth at most $1$.
For the converse inclusion we use a distributive law of addition over orbit finite unions:
\[
\bigcup_{i\in I} L_i \quad + \quad 
\bigcup_{j\in J} K_i \quad = \quad 
\bigcup_{\pair i j \in I\times J} L_i + K_j.
\]
Note that the Cartesian product $I\times J$ of orbit-finite sets $I$ and $J$ is necessarily orbit-finite
(cf.~\cite[Sect.~3]{atombook}).

Consider a rational set $X$ of data vectors of star-height $h\leq 1$.
If $h = 0$, by the distributive law the set $X$ is orbit-finite and hence vacuously semi-linear.
If $h=1$, by the distributive law we similarly deduce that, for every star subexpression $Y^*$,
the set $Y$ is orbit-finite; and moreover, the set $X$ is an orbit-finite union 
\begin{align} \label{eq:ofu}
X \quad = \quad \bigcup_{i\in I} X_i,
\end{align}
where each $X_i$ is a sum of star subexpressions $Y^*$ and orbit-finite sets.
As addition is commutative, preserves orbit-finiteness, and admits merging of stars:
\[
Y^* \ + \ Z^*  \quad = \quad \big(Y\cup Z\big)^*, 
%\big((X\cup\set{\varepsilon}) + (Y\cup\set{\varepsilon})\big)^*,
\]
each of sets $X_i$ is of the form
\[
Z + Y^*
\]
where $Y$ and $Z$ are both obit-finite.
Therefore each $X_i$ is semi-linear, and hence the orbit-finite~\eqref{eq:ofu} union is semi-linear too.
%
%Semilinear sets are closed addition:
%\[
%\bigcup_{i\in I} \, b_i + {P_i}^* \  + \ 
%\bigcup_{j\in J} \, c_j + {Q_i}^* \  = \!\!\!\!
%\bigcup_{(i,j)\in I\times J} \, (b_i + c_j) + {(P_i\cup Q_j)}^*.
%\]
%Closure under orbit-finite unions is easily shown using Lemma~\ref{lem:o-f}.
%Finally, closure under star follows by
%\[
%\Big(\bigcup_{i\in I} \, b_i + {P_i}^*\Big)^* \quad  = \quad .
%\]
%\todooo{is this abstract nonsense clear enough?}
\end{proof}

\medskip

\begin{proof}[\bf Proof of Lemma~\ref{lem:Marta}]
The proof is an adaptation of the argument from~\cite{Marta}.
Towards contradiction, suppose $\Par {L_3}$ is semilinear:
\begin{align*}
\Par{L_3}  \quad = \quad 
\bigcup_{i\in I} \, g_i + {P_i}^*.
\end{align*}
%
%As $L_3$ is equivariant, we can assume w.l.o.g.~that the indexing set $I$ and the function $i\mapsto (g_i,P_i)$
%are equivariant.
%Indeed, extending $I$ to its equivariant closure $\setof{\pi(i)}{\pi \in \Gr, i\in I}$, and likewise extending
%the indexing function, preserves $\Par {L_3}$ due to its equivariance.

For a data vector $p : \atoms\to\N$, let $\poj p = \size{\setof{a\in \atoms}{p(a)=1}}$ denote the number of atoms appearing exactly once in $p$.
%likewise we write $\poj w$ for data words over $\atoms$.
Our argument relies on a careful analysis of the limit value of the \emph{singularity ratio} $\frac {\poj p}{\size p}$, for $p \in \Par{L_3}$, 
when $\size p$ tends to infinity.

By induction on the length of $w\in L_3$ one easily proves:
\begin{claim} \label{claim:12}
Every $p\in\Par{L_3}$ satisfies $\poj{p} < \frac{1}{2} \size p$.
\end{claim}
%
%By another induction on length of $w\in L_3$ one easily proves:
%
\begin{claim} \label{claim:dom}
Every $p\in\Par{L_3}$ satisfies $\dom{p} \leq \frac{1}{2} \size p$.
\end{claim}
For a data vector $p : \atoms \to\N$ and $S\subseteq \atoms$,  we denote by $\bez p S$ the
data vector obtained from $p$ by removing all occurences of atoms from $S$:
\[
(\bez p S)(a) = \begin{cases} p(a) & \text{ if } a\notin S\\
0 & \text{ otherwise.} \end{cases}
\]
Let $S_i = \supp {P_i}$.
As a consequence of Claim~\ref{claim:12}, we get:
\begin{claim} \label{claim:<=12}
Every $p\in P_i$ ($i\in I$) satisfies $\poj{\bez p {S_i}} \leq \frac{1}{2} \size p$. % where $S = \supp {P_i}$.
\end{claim}
\noindent
\begin{proof}
Towards contradiction, suppose $\poj{\bez p {S_i}} > \frac{1}{2} \size p$ for some $p\in P_i$.
Let $S' = \dom{\bez p {S_i}}$.  % = \dom p - S$.
For an arbitrary permutation of atoms $\pi\in\Gr$ such that $\pi(a)=a$ for all $a\in S_i$,
%(call such permutations \emph{$S$-permutations}), 
we have $\pi(p) \in P_i$. 
Consider such permutations $\pi_1, \ldots, \pi_n \in \Gr$, such that $\pi_k(S')$ and $\pi_l(S')$ are disjoint for $k\neq l$.
As $\size {g_i}$ is fixed, for sufficiently large $n$ the vector
\[
g_i + \pi_1(p) + \ldots + \pi_n(p) \ \in \ \Par {L_3}
\]
%By adding the base $b_i$ and sufficiently many such permutations of $p$  we obtain a data vector in $\Par{L_3}$ that 
contradicts Claim~\ref{claim:12}. This completes the proof.
\end{proof}

We call a data vector $p:\atoms\to\N$ \emph{non-singular} if  $p(a)>1$ for some $a\in\atoms$.
Claim~\ref{claim:<=12} can be strengthened as long as non-singular data vectors are considered:
\begin{claim} \label{claim:<12}
Every $p\in P_i$ ($i\in I$) such that $\bez p {S_i}$ is non-singular, 
satisfies $\poj{\bez p {S_i}} < \frac{1}{2} \size p$.  %, where $S = \supp {P_i}$.
\end{claim}
\noindent
\begin{proof}
Indeed, suppose $\poj{\bez p {S_i}} \geq \frac{1}{2} \size p$ for some $p\in P_i$, 
and hence $\dom{\bez p {S_i}} > \frac 1 2 \size p$ due to non-singularity of $\bez p {S_i}$.
Considering similar permutations of $p$ as in the argument for Claim~\ref{claim:<=12},
we contradict Claim~\ref{claim:dom}.
\end{proof}

Let $k = \max\setof{\size{S_i}}{i\in I}$ be the maximal size of the support of $P_i$; note that $k$ is well defined as
the family of sets $\set{P_i}_{i\in I}$ is orbit-finite, and the size of the support is invariant inside an orbit.
Likewise, let  $t = \max\setof{\size{g_i}}{i\in I}$ be the maximal size of a base
and let $s = \max\setof{\size{p}}{p\in \bigcup_i P_i}$ be the maximal size of a period.
%Let $k = \max \set{l,s}$.

Let $Z  = \set{a_0, \ldots, a_k}\subseteq \atoms$ be some fixed $k+1$ atoms.
A word $v\in L_a$ (cf.~\eqref{eq:Marta}) we call \emph{varied} if all atoms different than $a$ appear at most once in $v$.
For every $m\in\N$ choose some arbitrary but fixed word $w_m\in L_3$ of the form:
\begin{align} \label{eq:wm}
w_m \ = \ v_0 \, v_1 \, \ldots \, v_{k} \ \in \ L_{a_0} \, L_{a_1} \, \ldots \, L_{a_{k}},
\end{align}
where each $v_i \in L_{a_i}$ is a varied word of length $2m$
and no atom appears in two distinct words $v_{a_i}$, $v_{a_j}$, for $i \neq j$.
%atoms $a_0, \ldots, a_k$ are different from other atoms
%appearing in $w_m$ and $\size {v_i} = 2m+2$.
%Let $Z_m  = \set{a_0, \ldots, a_k}$ and 
Let $q_m = \Par{w_m}$.
Hence $\size{w_m} = \size{q_m} = 2m(k+1)$.
As $\poj{q_m} = (m-1)(k+1)$, 
in the limit we have:  % $\poj{q_m} = \frac{1}{2} \size {q_m}$:
\begin{align} \label{eq:limit}
\lim_{m \to \infty} \frac{\poj{q_m}}{\size{q_m}}  \ = \  
\lim_{m \to \infty} \frac{(m-1)(k+1)}{2m(k+1)} \ = \ 
\frac{1}{2},
\end{align}
irrespectively of the choice of the words $w_m$.
Let $g_{i_m} + {P_{i_m}}^{\!\!*}$ ($i_m\in I$) be a linear set to which $q_m$  belongs.
Thus  
$q_m = g_{i_m} + p_m$, for $p_m \in {P_{i_m}}^{\!\!*}$.
% and hence:
%%
%\begin{align} \label{eq:limit}
%\lim_{m \to \infty} \frac{\poj{p_m}}{\size{p_m}}  = \frac{1}{2}
%\end{align}
%%
%since $\size{g_i}$ is bounded.
%Consider an arbitrary fixed word $w_m$ for a sufficiently large $m$ 
%Let $S_m = \supp {P_{i_m}}$.
Recalling~\eqref{eq:wm}, choose $a_{i_m} \in Z$ 
so that $a_{i_m}\notin S_{i_m}$ (such $a_{i_m}$ exists as $\size {S_{i_m}}\leq k$).
We split $p_m$:
\begin{align} \label{eq:p}
q_m \ = \ \, (g_{i_m} + p_{m,0}) \, + \, p_{m,1} \, + \, p_{m,>1},
\end{align}
where $p_{m,1}$ is a sum of vectors from $P_{i_m}$ that contain exactly one appearance of $a_{i_m}$;
$p_{m,>1}$ is a sum of vectors from $P_{i_m}$ that contain more than one appearance of $a_{i_m}$;
and  $p_{m,0}$ is a sum of vectors from $P_{i_m}$ that contain no appearance of $a_{i_m}$ at all.
Applying Claim~\ref{claim:12} to $g_{i_m} + p_{m,0} \in \Par{L_3}$, we obtain:
\begin{align} \label{eq:p0}
\limsup_{m \to \infty} \frac{\poj{g_{i_m} + p_{m,0}}}{\size{g_{i_m} + p_{m,0}}} \leq \frac{1}{2}.
\end{align}

Observe that the size of the sum of the last two data vectors in~\eqref{eq:p} constitutes, up to a constant $t$, at least
$\frac 1 {2(k+1)}$ fraction of the whole size $\size{p_m}$ (recall that $\size{g_i}$ is bounded by $t$):
\begin{align} \label{eq:1k}
\frac{1}{2(k+1)} \size{p_m}  - \size{g_i} \ <  \ \size{p_{m,1}  +  p_{m,>1}}
%  \leq  \frac{k}{2(k+1)} \size{p_m} + k  % \size{p_{m,1} \, + \, p_{m,>1} \, + \, p_{m,0}}
\end{align}
as it includes all $\frac{1}{2(k+1)} \size{p_m} + 1$ 
appearances of $a_{i_m}$ in $q_m$, except for at most $\size{g_i}$ many of them, possibly appearing in $g_i$.
% but also at most $k n$ atoms since $\size v \leq k$ for  $v\in P_{i_m}$. 
%From the right inequality we deduce
%$\lim_{m\to\infty} \size{p_{m,0}} = \infty$ and,
%as $\size {S_{i_m}}\leq k$ is bounded,  
%the size of $S_m$ is bounded by $k$, 
%
We are going to prove the following strict inequality
\begin{align} \label{eq:toprove}
\limsup_{m \to \infty} \frac{\poj{p_{m,1}+p_{m,>1}}}{\size{p_{m,1}+p_{m,>1}}} < \frac{1}{2}
\end{align}
which, together with inequalities~\eqref{eq:p0} and~\eqref{eq:1k},  implies 
$$\limsup_{m \to \infty} \frac{\poj{q_m}}{\size{q_m}}  < \frac{1}{2}$$ 
and thus contradicts the equality~\eqref{eq:limit}.
Call $p_{m,1}$ \emph{non-trivial} if it is a sum of at least two vectors from $P_m$.
When $p_{m,1}$ is trivial, $\size{p_{m,1}}\leq s$ is bounded and hence $p_{m,1}$ 
can be ignored in~\eqref{eq:toprove}.
We split the inequality~\eqref{eq:toprove} into two separate ones
\begin{align} \label{eq:toprove2}
\limsup_{m \to \infty} \frac{\poj{p_{m,1}}}{\size{p_{m,1}}} < \frac{1}{2}, \ 
\limsup_{m \to \infty} \frac{\poj{p_{m,>1}}}{\size{p_{m,>1}}} < \frac{1}{2}
\end{align}
and prove the first one assuming that $p_{m,1}$ is non-trivial, % for all sufficiently large $m$,
and the second one unconditionally.
This is enough to derive~\eqref{eq:toprove}.

Concerning the first inequality, 
%let
%$s = \max\setof{\size{p}}{p\in P}$ to be the maximal size of a period $p\in P$
%(which is well defined similarly as $k$).
we observe that the atom $a_{i_m}\notin S_{i_m}$ is counted in $\poj {\bez v {S_{i_m}}}$ for every data vector
$v$ contributing to the sum $p_{m,1}$, but if there are more than one of these vectors $v$, then
the atom $a_{i_m}$ is no more counted in $\poj {\bez {p_{m,1}} {S_{i_m}}}$.
Thus $\poj {\bez {p_{m,1}} {S_{i_m}}}$ loses, intuitively speaking, at least the $\frac{1}{s}$ fraction of the maximal possible value
$\frac{1}{2} \size{p_{m,1}}$ according to Claim~\ref{claim:<=12}.
This allows us to deduce:
\[
\frac{\poj{\bez{p_{m,1}} {S_{i_m}}}}{\size{p_{m,1}}} \leq \frac{1}{2} \cdot \Big(1-\frac{1}{s}\Big) < \frac 1 2
\]
which implies, in the limit, the first inequality in~\eqref{eq:toprove2},
as $\size{S_{i_m}}$ is  bounded  (by $k$).

Concerning the second inequality, let's put
\[
r = \max \setof{\frac{\poj{\bez v {S_i}}}{\size{v}}}{v \in P_i, \ \bez v {S_i} \text{ is non-singular} }.
\]
As before, $r$ is well defined due to orbit-finiteness of all $P_i$ and $(P_i)_{i\in I}$, and moreover 
$r < \frac 1 2$ by Claim~\ref{claim:<12}.
A crucial observation is that 
\[
\frac{\poj{\bez{p_{m,>1}} {S_{i_m}}}}{\size{p_{m,>1}}} \leq \frac{\poj{\bez v {S_{i_m}}}}{\size{v}}
\]
for some $v\in P_m$ that contributes to the sum $p_{m,>1}$, and hence
\[
\frac{\poj{\bez{p_{m,>1}} {S_{i_m}}}}{\size{p_{m,>1}}} \leq r < \frac 1 2
\]
which implies, in the limit, the second inequality in~\eqref{eq:toprove2},
as $\size{S_{i_m}}$ is bounded.
The inequalities~\eqref{eq:toprove2} are thus proved.
\end{proof}

\medskip

\subsection{Missing items in Section~\ref{sec:1CFG}}  \label{sec:1CFG-app} 

\begin{proof}[\bf Proof of Lemma~\ref{lem:h}]
%Production rules in $\Delta_2$ where the constraint $\varphi \equiv x = y = y'$ imposes all equalities we call 
%\emph{register-preserving}.
%Let $E \subseteq \Pi_2$ be the set of productions induced by register-preserving rules.
%
For a nonterminal $q\in Q$ and an atom $a\in\atoms$,
consider the set of derivation trees of $\G$
with root labeled by $\conf q a$, which use only
productions with the left-hand side in $Q\times \set{a}$
(thus every non-leaf in such a tree belongs to $Q\times\set{a}$), and where every leaf belongs either to
$H\times\atoms$ or to $Q\times (\atoms-\set{a})$.
Intuitively, we stop derivation at a terminal, or at a configuration with register value different than $a$
(i.e., at first register update along every path).
The language $L_{\conf q a}$ generated by such trees is obtained by applying a substitution to a
classical context-free language (with the finite set $Q\times\set{a}$ of nonterminals), and thus has rational Parikh image.
 
% consider the language $L_{\conf q a}$ of words generated by a derivation tree with root labeled with $\conf q a$,
% a grammar $\G_{\conf q a}$, with
%nonterminals $Q\times\set{a}$ and the initial nonterminal $\conf q a$, 
%obtained from  $\G$ by restricting to register-preserving production rules.
%As terminals of $\G_{\conf q a}$ we take $(H\cup Q)\times\set{a}$, which includes also the configurations $\conf p a$
%for $p\in Q$.

The proof is by induction on $n$. 
In case $n=0$, we observe that $\langh {\conf q a, 0}$ is the restriction of $L_{\conf q a}$ to terminals
$H\times\set{a}$:
\[
\langh {\conf q a, 0} \quad = \quad L_{\conf q a} \ \cap \ (H\times\set{a})^*
\]
and thus is itself a classical context-free language (with the finite set $Q\times\set{a}$ of nonterminals
and the finite set $H\times\set{a}$ of terminals); 
in consequence, it has rational Parikh image.

For the induction step we assume rationality of languages $\langh {\conf q a, n}$, and observe that
$\langh {\conf q a, {n+1}}$ is obtained by applying to the language $L_{\conf q a}$ the substitution:
\[
\conf {p} b \quad \mapsto \quad \langh {\conf p b, n}
\qquad
\pair h b \quad \mapsto \quad \pair h b,
\]
where $p\in Q$, $h \in H$, and $b\in\atoms$.
Indeed, intuitively speaking, $L_{\conf q a}$ allows for exactly one register update, while
$\langh {\conf p b, n}$ allows for ${\leq} n$ additional  register updates along every path.
Therefore $\langh {\conf q a, {n+1}}$ has rational Parikh image, as required.
\end{proof}

\medskip

\subsection{Missing items in Section~\ref{sec:anti}}  \label{sec:anti-app}
\label{sec:ap}

\begin{proof}[\bf Proof of Lemma~\ref{lem:minus}]
We transform a rational expression $E$ defining
a language $L\subseteq \Gamma^*$
into a rational expression $\widetilde E$ defining  $L \minu \alpha$.
We proceed by structural induction on $E$.
In case of orbit-finite union the transformation is distributive:
\[
\widetilde{\bigcup_{i\in I} E_i} \quad := \quad \bigcup_{i\in I} \widetilde{E_i}.
\]
In case of sum, the transformation is applied to one of summands:
\[
\widetilde{\ \ E_1 \ + \ E_2 \ \ } \ := \quad 
(\widetilde{E_1} \ + \ E_2) \quad \cup \quad
(E_1 \ + \ \widetilde{E_2}).
\]
In case of iteration, the transformation is applied to a single iteration
(which forces at least one iteration and hence rules out the vacuous generation of the zero vector $\zerovector$ due to $0$ iterations):
\[
\widetilde{\ E^* \ } \quad := \quad
\widetilde{E\ }  E^*.
\]
Finally, the induction base, for a singleton $\set{\beta}$, is given by:
\[
\widetilde{\set{\beta}} \quad := \quad \begin{cases}
\zerovector & \text{ if } \ \beta = \alpha \\
\emptyset & \text{ otherwise}.
\end{cases}
\]
\end{proof}

\medskip

For a word $w = \alpha_1 \ldots \alpha_n \in\Gamma^*$ we denote by $\srces{w}$ the sequence
$\src {\alpha_1} \ldots \src {\alpha_n}$ of sources.
For a finite subset $X\subset \atoms$  and a regular language $K \subseteq X^*$ 
we define:
\[
\apaX K a \ = \  % \acall \ \cap \ \setof{\alpha\in\Sigma}{\src \alpha \in X}^*.
\setof{w\in\Gamma^*}{\srces w \in K, a\notin\trg w}.
\]
%Clearly $\apaX K a \ \subseteq \ \setof{\alpha\in\Sigma}{\src \alpha \in X}^*$.
%
\begin{lemma} \label{lem:X}
For every finite set $X\subset\atoms$ and regular language $K\subseteq X^*$, 
the languages $\apaX K a$ are rational.
\end{lemma} 
\begin{proof}
Consider the finite set $\Delta = (X\cup\set{a})^{2}$ as an alphabet,
and the regular language $P \subseteq \Delta^*$
of all \emph{$K,a$-paths}, i.e., all nonempty sequences 
\[
\tuple{d_1, d_2}\, \tuple{d_2, d_3} \, \ldots\, \tuple{d_{n}, d_{n+1}} \ \in \ \Delta^*
\]
such that 
$d_1 d_2 \ldots d_n \in K$ and 
$d_{n+1} = a$.
The language $\apaX X a$ is obtained from $P$ by the substitution
\[
\tuple{d,e} \quad \mapsto \quad
\bigcup_{\set{e',e''}\in\pow 2 {\atoms-\set{e}}} \tuple{d,\set{e',e''}}
\]
and is thus rational.
\end{proof}

\medskip

%\subsubsection{Lemma~\ref{lem:D} implies Lemma~\ref{lem:C}} 
%
\begin{proof}[\bf Proof of Lemma~\ref{lem:DC}]

We show that rationality of $\Par \acall$ implies rationality of $\Par \acallgamma$.
To this aim we define, for distinct atoms $b, c\in\atoms$, the language
\begin{align*}
%L_{a, \neq a} \ := \ & \bigcup_{b\in \atoms - \set{a}} (a,b) \\
%L_{\neq a, a} \ := \ & \bigcup_{b\in \atoms - \set{a}} (b,a) \\
K_{b c} \ := \ & \  \apaX {b \set{b,c}^*} b \,\, \tuple{b,\set{b,c}} \quad \cup \quad 
\apaX {b \set{b,c}^*} c \,\, \tuple{c,\set{b,c}} 
\end{align*}
of all anti-paths where the last target is $\set{b,c}$, all sources are in $\set{b,c}$, and the first one is $b$.
Languages $K_{b c}$ are rational, due to Lemma~\ref{lem:X}.
Further, for pairwise distinct atoms $a, b, c$ we define the following rational language
\begin{align*}
%K_{b c} \ := \ & \  \Big(\bigcup_{b',c'\neq b} \tuple{b, \set{b',c'}}\Big)^* \, \tuple{b,\set{b,c}} \\
K_{a \set{b,c}} \ := \ & \  \tuple{a,\set{b,c}} \ \ \cup \ \bigcup_{b',c' \in\atoms{-}\set{b}} \tuple{a,\set{b',c'}} \, K_{b c}. % \ \cup \  (a,c) \, L_b.
\end{align*}
Note that 
the source of the first letter in every word in $K_{a \set{b,c}}$  is $a$, and the target of the last letter
is $\set{b,c}$.
Lemma~\ref{lem:DC} follows once we show the following claim:
\begin{claim}
$\Par \acallgamma$ is obtained from $\Par \acall$ by applying \emph{twice} the substitution
\[
\tuple{d, \set{e,f}} \quad \mapsto \quad K_{d \set{e,f}}.
\]
\end{claim}
\noindent
(We consider Parikh images of $\acallgamma$ and $\acall$, instead of the languages themselves,
only because we reason below up to cyclic shifts.)
%
%\begin{proof}
From now on we concentrate on the proof of the claim.
Let $\widetilde \acall$ denote the set of data vectors obtained from $\Par \acall$ by
applying twice the above-defined substitution.
By the very definition, $\widetilde \acall \subseteq \Par \acallgamma$.
For the converse inclusion, we prove that every data vector $v \in \Par \acallgamma$ belongs to $\widetilde \acall$.
%by induction on the number of unwanted letters $\tuple{b', \set{b',c'}}\in \Gamma - \Sigma$ in $v$.

%\parainproof{Base case} 
If $v$ contains no unwanted letters from $\Gamma - \Sigma$ then $v\in \Par \acall$,
and the claim follows due to $\Par \acall \subseteq \widetilde \acall$.

%\parainproof{Induction step}
Otherwise, choose an anti-cycle $w\in\acallgamma$ with $v = \Par w$ and 
consider the last  appearance of an unwanted letter $(b, \set{b,c}) \in \Gamma-\Sigma$ in $w$.
Applying a cyclic shift ($\rightarrow$) we can assume, w.l.o.g., that the letter is the last one in $w$.
Let $u$ be the maximal suffix of $w$
that belongs to $K_{b c}$ (or, symmetrically, to $K_{c b}$):
\[
w \quad = \quad w' \, u.
\]
We observe that $w'\neq \emptyword$; indeed, as $\src u = b \in \trg{u} = \set{b,c}$,
the word $u$ itself is not an anti-cycle.

Let $w' = w'' \tuple{a, \set{b', c'}}$; since $w$ is an anti-chain we have $b\notin \set{b', c'}$, and 
by maximality of $u$ we have $a\notin\set{b,c}$.
Then $u' = \tuple{a, \set{b', c'}} \, u \in K_{a, \set{b,c}}$.
%and  in consequence $\tuple{a', \set{b'', c''}} \, v \, u \in  \klp {b'} {c'} a$
Replace the suffix $u'$ by $\tuple{a,\set{b,c}}$, thus obtaining a data word 
$\widetilde w = w'' \tuple{a,\set{b,c}}$ with smaller number of occurrences of unwanted letters. 
We continue in the same way with $\widetilde w$ untill all occurences of letters from $\Gamma - \Sigma$ are eliminated.
A crucial observation is that during elimination of all letters, except for possibly the very last one, 
the total sum of cyclic shifts  ($\rightarrow$) performed does not exceed the full cyclic shift of $w$.
Therefore, Parikh image of the word obtained by elimination of all unwanted letters except for the last one,
belongs to the result of application the substitution once to $\acall$.
In consequence, the final word belongs to the result of applying the substitution twice, as required.
\end{proof}

\medskip

\begin{proof}[\bf Proof of Lemma~\ref{lem:belown}]

%For a word $w = \alpha_1 \ldots \alpha_n \in\Gamma^*$ we denote by $\srces{w}$ the sequence
%$\src {\alpha_1} \ldots \src {\alpha_n}$ of sources.
For a finite subset $X\subseteq \atoms$ 
% a regular language $K \subset X^*$ 
the language
\[
\acallgamma_X \ = \  \acallgamma \ \cap \ \setof{\alpha\in\Gamma}{\src \alpha \in X}^*.
%\setof{w\in\Gamma^*}{\srces w \in K, a\notin\trg w}
\]
%Clearly $\apaX K a \ \subseteq \ \setof{\alpha\in\Sigma}{\src \alpha \in X}^*$.
%
is rational, due to Lemma~\ref{lem:X}, as it equals
\[
\bigcup_{a\in X} \apaX {a X^*} a,
\]
and hence so is its restriction $\acall_X = \acallgamma_X \cap \Sigma^*$.
The language $\acallpar{<n}$, being the union  of all the rational languages $\acall_X$
for subsets $X\subseteq \atoms$ of cardinality $< n$, is thus rational as well.
\end{proof}

\medskip

\begin{proof}[\bf Proof of Lemma~\ref{lem:nondegrat}]

%Fix $n\in\N$ sufficiently large for Lemma~\ref{lem:above}.
Fix $n\geq 6$.
We define \emph{the kernel} of a data vector $v : \Sigma\to\N$ as the intersection of all targets in $v$:
\[
\kerne v \ = \ \bigcap_{\alpha \in \dom v} \trg \alpha.
\] 
The size of the kernel is $0, 1$ or $2$.
For $X\subseteq \atoms$ of size at most $2$,
let 
\[
\ndegordgre X n \ = \ \setof{v\in\ndegn n}{\kerne v = X}.
\]
As $\ndegn n = \bigcup_{X} \ndegordgre X n$, it is enough to show that the sets
$\ndegordgre X n$ are rational.
	This, in turn, is implied by the following decomposition property
	% \todo{with respect to which order? do not use things thatare not defined.} 
	of sets $\ndegordgre X n $:
\begin{align} \label{eq:n}
\ndegordgre X n \ = \ \ndegordeq X n \ + \ \Par {{\Sigma_X}^*},
\end{align}
where $\ndegordeq X n = \setof{v\in\ndegordgre X n}{\ord v = n}$ and
$\Sigma_X = \setof{\alpha\in\Sigma}{X\subseteq \trg \alpha}$. 
Towards showing the decomposition~\eqref{eq:n} we prove that kernel-preserving extensions 
by one letter $\alpha \in\Sigma$ preserve membership in $\ndegall$: 
\[
v \in \ndegall, \ \kerne v = \kerne {v + \alpha} \implies v + \alpha \in \ndegall;
\]
	and also that 
	there always exists a letter $\alpha$ that one can remove from a vector in $\ndegn {n+1}$,
preserving kernel and membership in $\ndegall$:
%for every
%$v \in \ndegn n$ there is some $\alpha \in \dom v$ such that  
%$\kerne v = \kerne {v - \alpha}$ and
%$v - \alpha \in \ndegall$.
\begin{align*}
v \in \ndegn {n+1} \implies \prettyexists{\alpha \in \dom v\ }{
\ & \kerne v = \kerne {v - \alpha},} \\
\ & v - \alpha \in \ndegall.
\end{align*}

Concerning the first property,
suppose $v \in \ndegall$ and $\kerne v = \kerne {v + \alpha}$.
We thus know that $v$ satisfies conditions (1)--(3)  and that
$d = \src \alpha \notin \kerne v$ since $d\notin \trg \alpha$.
This implies that $v + \alpha$ satisfies (1).
For conditions (2)--(3) we consider two separate cases.
If $d \in \srcs v$ then adding $\alpha$ may only increase in-neighbour sets $\inn \_$
and preceeding-letter sets $\precn {\_}$,
and hence $v+\alpha$ satisfies (2)--(3).
Otherwise, suppose $d\notin \srcs v$ is a fresh source.
We reason by contradiction.
If $v + \alpha$ violates (2) for $d$ and some $e\in\srcs v$, then $v$ necessarily violates (1) due to $\inn e = \emptyset$.
If $v + \alpha$ violates (3) for $d$ and some $e\in\srcs v$, then all $\beta \in \dom v$, except for exactly one, satisfy
$\trg \beta = \set{d,e}$ and hence forcedly $\src \beta \neq e$. 
Therefore there is exactly one $e$-sourced letter in $v$
	and $\inn e = \emptyset$, and hence $v$ violates (1) again.
%	\todo{I think it should be clarified here, I don't see whay $\inn e =\emptyset$.}  
%\todooo{SL: rozwinalem ten paragraf}

We now concentrate on the second property.
Removal of a letter from $v$ may only increase (inclusion-wise) the kernel,
say from $X$ to $X'$, but this only happens if $v(\alpha) = 1$, $X'\not\subseteq \trg \alpha$, 
and $X'\subseteq \trg \beta$ for all $\beta \in\dom v - \set{\alpha}$.
By inspection of possible sizes 1, 2 of $X'$, one deduces that $v$ may contain
at most two such kernel-increasing letters.
This eliminates at most $2$ potential sources $\src \alpha$.

Non-degeneracy can be only violated by vertices in the source graph of in-degree below $2$. 
Therefore non-degeneracy of  $v-\alpha$ is guaranteed if removal of $\alpha$ does not decrease
in-degree of any vertex below $2$, i.e., $\src \alpha$ does not belong to $\inn d$ 
for $d\in\srcs v$ of in-degree  $\ind d \leq 2$.
For sufficiently large $n$, similarly as in Claim~\ref{claim:ind},
there are at most $2$ such vertices $d$ in $\srcs v$.
This eliminates at most $4$ potential sources
$\src \alpha$. 
%\todo{I think it is too big jump, you should explain using definition of non-degeneracy.}
%\todooo{SL: rozwinalem ten paragraf}

In total, at most $6$ potential sources $\src \alpha$ are eliminated.
Therefore, as long as $\ord v>6$, there is $\alpha \in \dom v$ such that
$\kerne v = \kerne {v - \alpha}$ and
$v - \alpha \in \ndegall$.
\end{proof}

%
%\begin{lemma} \label{lem:X}
%For every finite $X\subset\atoms$, the languages $\apaX X a$ are rational.
%\end{lemma} 
%%
%\begin{proof}
%Consider the finite set $\Delta = X^{(2)}$ as an alphabet,
%and the regular language $P \subseteq \Delta^*$
%of all \emph{$X$-paths}, i.e., all nonempty sequences 
%\[
%\tuple{d_1, d_2}\, \tuple{d_2, d_3} \, \ldots\, \tuple{d_{n}, d_{n+1}} \ \in \ \Delta^*
%\]
%such that 
%%$d_1 d_2 \ldots d_n \in K$ and 
%$d_{n+1} = d_1$.
%The language $\apaX X a$ is obtained from $P$ by the substitution
%\[
%\tuple{d,e} \quad \mapsto \quad
%\bigcup_{\set{e',e''}\in\pow 2 {\atoms-\set{d,e}}} \tuple{d,\set{e',e''}}
%\]
%and is thus rational.
%\end{proof}
%

%\input{app}
%\input{effective}
%\input{ATTIC/antipaths}

\end{document}